\documentclass[accepted=2025-05-08,onecolumn,a4paper,11pt]{quantumarticle}
\pdfoutput=1

\usepackage[utf8]{inputenc}
\usepackage[english]{babel}
\usepackage[T1]{fontenc}
\usepackage{amsmath, amsthm, amssymb, mathtools, xcolor, framed, float}
\usepackage{caption, comment}

\usepackage{tikz}
\usetikzlibrary{patterns, patterns.meta, decorations.pathreplacing, calligraphy, shapes, intersections, positioning}
\tikzstyle{comm}=[ellipse,draw=black,fill=gray!5]
\tikzstyle{leaf}=[]

\usepackage{hyperref}
\hypersetup{
    colorlinks=true,
    linkcolor=blue,
    filecolor=magenta,
    urlcolor=cyan,
}

\usepackage{algorithm}
\usepackage[noend]{algpseudocode}

\allowdisplaybreaks

\newcommand\blfootnote[1]{%
    \begingroup
    \renewcommand\thefootnote{}\footnote{#1}%
    \addtocounter{footnote}{-1}%
    \endgroup%
}

\newcommand{\ket}[1]{|#1\rangle}
\newcommand{\bra}[1]{\langle#1|}

\newcommand{\braket}[2]{\langle#1|#2\rangle}
\newcommand{\norm}[1]{\left\lVert#1\right\rVert}
\newcommand{\spanProgramP}{P}
\newcommand{\spann}{\mathrm{span}}
\newcommand{\posWitX}{w_x}
\newcommand{\negWitX}{\overline{w}_x}
\newcommand{\wsizePos}{\textnormal{wsize}^+}
\newcommand{\wsizeNeg}{\textnormal{wsize}^-}
\newcommand{\wsize}{\textnormal{wsize}}
\newcommand{\dtsize}{\textnormal{DTSize}}
\newcommand{\rdtsize}{\textnormal{RDTSize}}
\newcommand{\Adv}{\textnormal{Adv}}

\newcommand{\score}{\mathsf{score}}
\newcommand{\parent}{\mathsf{parent}}
\newcommand{\pparent}{\mathsf{pparent}}
\newcommand{\contract}{\mathsf{contract}}
\newcommand{\update}{\mathsf{update}}
\newcommand{\llabel}{\mathsf{label}}
\newcommand{\child}{\mathsf{child}}

\newtheorem{theorem}{Theorem}[section]
\newtheorem{corollary}[theorem]{Corollary}

\newtheorem{lemma}[theorem]{Lemma}
\newtheorem{claim}[theorem]{Claim}

\newtheorem{observation}[theorem]{Observation}
\newtheorem{definition}[theorem]{Definition}

\newcommand{\eps}{\varepsilon}
\renewcommand{\epsilon}{\varepsilon}
\newcommand{\cbra}[1]{\left\{#1\right\}}
\newcommand{\rbra}[1]{\left(#1\right)}

\newcommand{\mathify}[1]{\ifmmode{#1}\else\mbox{$#1$}\fi}

\newcommand{\zone}{\{0, 1\}}

\newcommand{\rank}{\mathsf{rank}}
\newcommand{\leaf}{\mathsf{leaf}}
\renewcommand{\root}{\mathsf{root}}
\renewcommand{\index}{\mathsf{index}}
\newcommand{\rrank}{\mathsf{rrank}}
\newcommand{\val}{\mathsf{val}}

\DeclarePairedDelimiter\abs{\lvert}{\rvert}%

\makeatletter
\let\oldabs\abs
\def\abs{\@ifstar{\oldabs}{\oldabs*}}
\makeatother

\newcommand{\AND}{\mathsf{AND}}

\newcommand{\OR}{\mathsf{OR}}

\newcommand{\leaff}{\tilde{f}}
\newcommand{\ZO}{\{0,1\}}

\newcommand{\A}{\mathcal{A}}
\renewcommand{\P}{\mathcal{P}}

\newcommand{\T}{\mathcal{T}}

\newcommand{\calR}{\mathcal{R}}
\newcommand{\G}{\mathcal{G}}

\newcommand{\C}{\mathbb{C}}
\newcommand{\N}{\mathbb{N}}

\newcommand{\Q}{\mathsf{Q}}
\newcommand{\DT}{\mathsf{D}}
\newcommand{\RDT}{\mathsf{R}}
\newcommand{\OPT}{\textnormal{OPT}}

\title{Improved Quantum Query Upper Bounds Based on Classical Decision Trees}
\author{Arjan Cornelissen}
\email{ajcornelissen@outlook.com}
\affil{Simons Institute for the Theory of Computing, University of California, Berkeley, United States of America}
\author{Nikhil S.~Mande}
\email{mande@liverpool.ac.uk}
\affil{University of Liverpool, United Kingdom}
\author{Subhasree Patro}
\email{patrofied@gmail.com}
\affil{Eindhoven University of Technology, Netherlands}

\begin{document}

    \maketitle\blfootnote{An extended abstract of this work appeared in the proceedings of FSTTCS 2022~\cite{CMP22}.}\vspace{-1em}

    \begin{abstract}
        We consider the following question in query complexity: Given a classical query algorithm in the form of a decision tree, when does there exist a quantum query algorithm with a speed-up (i.e., that makes fewer queries) over the classical one?
        We provide a general construction based on the structure of the underlying decision tree, and prove that this can give us an up-to-quadratic quantum speed-up in the number of queries.
        In particular, our results give a bounded-error quantum query algorithm of cost $O(\sqrt{s})$ to compute a Boolean function (more generally, a relation) that can be computed by a classical (even randomized) decision tree of size $s$. This recovers an $O(\sqrt{n})$ algorithm for the Search problem, for example.

        Lin and Lin [Theory of Computing'16] and Beigi and Taghavi [Quantum'20] showed results of a similar flavor. Their upper bounds are in terms of a quantity which we call the ``guessing complexity'' of a decision tree. We identify that the guessing complexity of a decision tree equals its \emph{rank}, a notion introduced by Ehrenfeucht and Haussler [Information and Computation'89] in the context of learning theory. This answers a question posed by Lin and Lin, who asked whether the guessing complexity of a decision tree is related to any measure studied in classical complexity theory. We also show a polynomial separation between rank and its natural randomized analog for the complete binary AND-OR tree.

        Beigi and Taghavi constructed span programs and dual adversary solutions for Boolean functions given classical decision trees computing them and an assignment of non-negative weights to edges of the tree. We explore the effect of changing these weights on the resulting span program complexity and objective value of the dual adversary bound, and capture the best possible weighting scheme by an optimization program. We exhibit a solution to this program and argue its optimality from first principles. We also exhibit decision trees for which our bounds are asymptotically stronger than those of Lin and Lin, and Beigi and Taghavi. This answers a question of Beigi and Taghavi, who asked whether different weighting schemes in their construction could yield better upper bounds.
    \end{abstract}

    \newpage

    \tableofcontents

    \section{Introduction}
    In this paper, we address the following question: given a classical algorithm performing a task, along with a description of the algorithm, when can one turn it into a quantum algorithm and obtain a speed-up in the process?
    Of specific interest to us in this paper is the case when the classical algorithm can be efficiently represented by a decision tree. For instance, consider the problem of identifying the first marked item in a list, say $x$, of $n$ items, each of which may or may not be marked.  The input is initially unknown, and one has access to it via an \emph{oracle}. On being queried $i \in [n]$, the oracle returns whether or not $x_i$ is marked. Moreover, queries can be made in superposition.
    The goal is to minimize the number of queries in the worst case and output the correct answer with high probability for every possible input list. This problem is closely related to the Search problem, and is known to admit a quadratic quantum speed-up (its classical query complexity is $\Theta(n)$ and quantum query complexity is $\Theta(\sqrt{n})$)~\cite{Gro96, DH96, Kot14, LL16}. It is easy to construct a classical decision tree (in fact, a decision list) of depth $n$ and size (which is the number of nodes in the tree) $2n + 1$ that solves the above-mentioned problem. In view of the quadratic speed-up that quantum query algorithms can achieve for this problem, this raises the following natural question:
    Given a classical decision tree of size $s$ that computes a function, is there a bounded-error quantum query algorithm of cost $O(\sqrt{s})$ that solves the same function? Among other results, we answer this in the affirmative (see Corollary~\ref{cor:ubs}).
    We obtain our quantum query upper bounds by constructing explicit \emph{span programs} and \emph{dual adversary solutions} by exploiting the structure of the initial classical decision tree. In the discussions below, let $\Q_\eps(f)$ be the $\eps$-error quantum query complexity of $f$. When $\eps = 1/3$, we drop the subscript and call $\Q(f)$ the bounded-error quantum query complexity of $f$.

    \subsection{Span Programs}

    Span programs are a computational model introduced by Karchmer and Wigderson~\cite{KW93}. Roughly speaking, a span program defines a function depending on whether or not the ``target vector'' of an input is in the span of its associated ``input vectors''. Span programs were first used in the context of quantum query complexity by Reichardt and \v{S}palek~\cite{RS12}, and it is known that span programs characterize bounded-error quantum query complexity of Boolean functions up to a constant factor~\cite{Rei11, LMRSS11}. Span programs have been used to design quantum algorithms for various graph problems such as \textit{st}-connectivity~\cite{BR12}, cycle detection and bipartiteness testing~\cite{Ari15, CMB18}, graph connectivity~\cite{JJKP18}, and has been also used for problems such as
    formula evaluation~\cite{RS12, Rei09, JK17}. Recently, Beigi and Taghavi~\cite{BT19} defined a variant of span programs (non-binary span programs with orthogonal inputs, abbreviated NBSPwOI) for showing upper bounds on the quantum query complexity of non-Boolean input/output functions $f:[\ell]^n\rightarrow [m]$. Moreover in a follow-up work~\cite{BT20}, they also use non-binary span programs for showing upper bounds for a variety of graph problems, for example, the maximum bipartite matching problem.

    \subsection{Dual Adversary Bound}
    The general adversary bound for Boolean functions $f$ (Equation~\eqref{eq:GeneralisedAdversarySolution}) was developed by H\o yer, Lee and \v{S}palek~\cite{HLS07}, and they showed that this quantity gives a lower bound on the bounded-error quantum query complexity of $f$. It was eventually shown that this bound actually characterizes the bounded-error quantum query complexity of $f$ up to a constant factor~\cite{Rei11, LMRSS11}. Quantum query algorithms have been developed by explicitly constructing dual adversary solutions, for example for the well-studied $k$-distinctness problem~\cite{Bel12}, \textit{S}-isomorphism and hidden subgroup problems~\cite{Bel19}, gapped group testing~\cite{ABRW16}, etc. The general adversary bound can be expressed as a semidefinite program that admits a dual formulation. Thus, feasible solutions to the dual adversary program yield quantum query upper bounds.

    \subsection{Related Works}
    Lin and Lin~\cite{LL16} showed how a classical algorithm computing a function $f:D_f \rightarrow \mathcal{R}$ with $D_f \subseteq \ZO^n$, equipped with an efficient `guessing scheme', could be used to construct a faster quantum query algorithm computing $f$. Moreover, their results apply to the setting where $f \subseteq \zone^n \times \mathcal{R}$ is a relation. A deterministic algorithm computing a relation $f$ takes an input $x \in \zone^n$ and outputs a value $b$ such that $(x,b) \in f$. More precisely, they showed the following:\footnote{For ease of readability, we state the theorem for deterministic decision trees. Lin and Lin actually showed a stronger statement that holds even if one considers randomized decision trees and randomized guessing algorithms.}

    \begin{theorem}[Lin and Lin~{\cite[Theorem 5.4]{LL16}}]\label{thm: linlin}
        Let $f \subseteq \zone^n \times \calR$ be a relation. Let $\A$ be a deterministic algorithm computing $f$ that makes at most $T$ queries. Let $x_{p_1}, \dots, x_{p_{\tilde{T}(x)}}$ be the query results of $\A$ on an a priori unknown input string $x$. For $x \in \zone^n$, let $\tilde{T}(x) \leq T$ denote the number of queries that $\A$ makes on input $x$. Suppose there is another deterministic algorithm $\G$ which takes as input $x_{p_1}, \dots, x_{p_{t - 1}} \in \zone^{t-1}$ for any $t \in [T]$, and outputs a guess for the next query result of $\A$. Assume that $\G$ makes at most $G$ mistakes for all $x$. That is,
        \[
        \forall x \in \zone^n, \quad \sum_{t = 1}^{\tilde{T}(x)} \abs{\G(x_{p_1}, \dots, x_{p_{t - 1}}) - x_{p_t}} \leq G.
        \]
        Then,
        \[
        \Q(f) = O(\sqrt{TG}).
        \]
    \end{theorem}

    Lin and Lin provided two proofs of the above theorem, one of which involved constructing an explicit quantum query algorithm. This algorithm is iterative, and works as follows:
    In the $i$'th iteration, use a modified version of Grover's search algorithm to find the next mistake that $\G$ makes. Since the algorithm uses at most $T$ queries and makes at most $G$ mistakes on any input, the quantum query complexity of the final algorithm can be shown to be $O(\sqrt{TG})$ using the Cauchy-Schwarz inequality.\footnote{We skip some crucial details here, such as the cost of the modified Grover's search algorithm, and an essential step that uses span programs to avoid a logarithmic overhead that error reduction would have incurred. Techniques that address both of the above-mentioned steps are due to Kothari~\cite{Kot14}.}

    Recently, Beigi and Taghavi~\cite{BT20} gave an alternate proof of Theorem~\ref{thm: linlin} using the framework of \emph{non-binary span programs} introduced by them in~\cite{BT19}. Using this they showed that a similar statement also holds for functions $f : [\ell]^n \to \calR$ with non-binary inputs and outputs.
    A sketch of their proof is as follows: Given an assignment of real weights to the edges of the given decision tree computing $f$, they construct a dual adversary solution and span program. The complexity of the resultant query algorithm is a function of the weights assigned to the edges. They propose a particular weighting scheme based on an edge-coloring (guessing algorithm) of the given decision tree. For Boolean functions, the quantum query complexity upper bound they obtain matches the bound of Lin and Lin (Theorem~\ref{thm: linlin}), which is $O(\sqrt{TG})$, where $T$ denotes the depth and $G$ the \emph{guessing complexity} (see Definition~\ref{def:guessing-complexity}) of the underlying decision tree, respectively.

    In a follow-up work~\cite{BTT21}, conditional on some constraints, Beigi, Taghavi and Tajdini implement the span-program-based algorithm of~\cite{BT20} in a time-efficient manner. Recently, Taghavi~\cite{Tag22} used non-binary span programs to give a tight quantum query algorithm for the oracle identification problem, simplifying an earlier algorithm due to Kothari~\cite{Kot14}.

    \subsection{Our Contributions}

    We observe here that the guessing complexity of a deterministic decision tree equals its \emph{rank} (see Definition~\ref{defn: rank} and Claim~\ref{claim: gcoloring equals rank}), which is a combinatorial measure introduced by Ehrenfeucht and Haussler~\cite{EH89} in the context of learning theory. 
    This answers a question of Lin and Lin~\cite[page 4]{LL16}, where the authors asked if $G$ is related to any combinatorial measure studied in classical decision-tree complexity. The rank of a function (which is the minimum rank of a decision tree computing it) can be unboundedly smaller than the function's certificate complexity, sensitivity, block sensitivity, and even (exact or approximate) polynomial degree, as can be easily witnessed by the $\OR$ function. See, for example,~\cite{DM21} for more relationships between rank and other combinatorial measures of Boolean functions. In view of the above-mentioned equivalence of $G$ and decision tree rank, Theorem~\ref{thm: linlin} has a clean equivalent formulation as follows.
    \begin{theorem}\label{thm: qq upper bound in terms of rank}
        Let $\T$ be a decision tree computing a relation $f \subseteq \zone^n \times \calR$ with depth $T$ and rank $G$. Then,
        \[
        \Q(f) = O(\sqrt{TG}).
        \]
    \end{theorem}
    We introduce a new measure, \emph{randomized rank} of randomized decision trees (see Definition~\ref{defn: rank}), which we denote by $\rrank$ and which is a natural probabilistic analog of rank. This exactly captures the notion of the randomized analog of the value of $G$ in Theorem~\ref{thm: linlin}.  It is easy to show with this definition that our proof of Theorem~\ref{thm: qq upper bound in terms of rank} also holds with respect to randomized decision trees and randomized rank. Thus we obtain the following easy-to-state reformulation of~\cite[Theorem~5.4]{LL16}.
    \begin{theorem}\label{thm: qq upper bound in terms of randomized rank}
        Let $\T$ be a randomized decision tree computing a relation $f \subseteq \zone^n \times \calR$ with depth $T$ and randomized rank $G$. Then,
        \[
        \Q_{2/5}(f) = O(\sqrt{TG}).
        \]
    \end{theorem}
    Thus, up to a small loss in the success probability, upper bounds obtained from Theorem~\ref{thm: qq upper bound in terms of randomized rank} are strictly stronger than those obtained from Theorem~\ref{thm: qq upper bound in terms of rank} for relations whose randomized rank is much smaller than their rank. Hence a natural question is if this can be the case.

    It is easy to exhibit maximal separations between rank and randomized rank for partial functions. One such separation is witnessed by the Approximate-Majority function, which is the function that outputs 0 if the Hamming weight of an $n$-bit input is less than $n/3$ and outputs 1 if the Hamming weight of the input is more than $2n/3$. It is easy to show an $\Omega(n)$ lower bound on its rank and an $O(1)$ upper bound on its randomized rank (see Claim~\ref{claim: apxmaj rank rrank}). Whether or not such a separation holds for total Boolean functions is not so clear. We show a polynomial separation for the complete binary AND-OR tree. This is the first of our main theorems.

    \begin{theorem}\label{thm: rank rrank separation nand tree}
        For the complete binary AND-OR tree $F : \zone^n \to \zone$,
        \[\rrank(F) = O\rbra{\rank(F)^{\log\frac{1 + \sqrt{33}}{4}}} \approx O\rbra{\rank(F)^{.753\dots}}.\]
    \end{theorem}
    To prove this theorem, we show a rank lower bound of $\Omega(n)$ using known connections between rank and Prover-Delayer games~\cite{PI00}, and the randomized-rank upper bound immediately follows from an upper bound of $O\rbra{n^{\log\frac{1 + \sqrt{33}}{4}}}$ on its randomized decision-tree complexity~\cite{SW86}.

    Beigi and Taghavi~\cite[Section 6]{BT20} asked if one could improve their results by using different choices of weights in their constructions of span programs and dual adversary vectors.
    For decision trees that make queries to a bit string, we answer this question in the affirmative by providing a weighting scheme that improves upon their bounds. We exhibit a family of decision trees for which our bounds are strictly stronger than those of Beigi and Taghavi, and Lin and Lin (see Corollary~\ref{cor:ubs} and the discussion below it). We argue that optimizing the dual adversary bound and the witness complexity of Beigi and Taghavi's span program with variable weights is a minimization program with constraints linear in the variables and inverses of the variables. We give a weighting scheme and moreover we show in Theorem~\ref{thm:weight-optimality} that our weighting scheme is optimal, thus subsuming Theorem~\ref{thm: linlin}.

    Beigi and Taghavi's construction also works for decision trees that query a non-Boolean string. They adapt the weighting scheme from the Boolean case by associating one weight to a single query outcome, and the other weight to all the other query outcomes. One can adapt the weighting scheme we introduce in this work to the non-Boolean setting in a similar manner, however we do not prove optimality of the weighting scheme in this case. We leave a possibly tighter generalization to the non-Boolean setting to future work.


    For a relation $f \subseteq \zone^n \times \calR$ and a deterministic decision tree $\T$ computing it, let $\leaff$ be the function that takes an $n$-bit string $x$ as input, and outputs the leaf of $\T$ reached on input $x$. Let $P(\T)$ denote the set of root-to-leaf paths in $\T$. For a path $P \in P(\T)$, let $\overline{P}$ denote the set of edges that have exactly one vertex in common with $P$.
    \begin{definition}[Weight optimization program]
        \label{def:weight-program}
        For a decision tree $\T$, define its \emph{weight optimization program} by the minimization problem with constraints outlined in Program~\ref{program: weights}. Let $\OPT_{\T}$ denote the optimum of this program.
    \end{definition}

    \begin{table}[H]
        \centering~
        \begin{tabular}{|llllll|}
            \hline
            Variables & $\cbra{W_e: e~\textnormal{is an edge in}~\T}, \alpha, \beta$ & & & & \\
            Minimize  & $\sqrt{\alpha\beta}$ & & & & \\
            s.t. & $\sum_{e \in \overline{P}}W_e$ & $\leq \alpha,$ & & for all paths $P \in P(\T)$ & \\
            & $\sum_{e \in P}\frac{1}{W_e}$ & $\leq \beta,$ & & for all paths $P \in P(\T)$ & \\
            & $W_e$ & $> 0,$ & & for all edges $e$ in $\T$ & \\
            & $\alpha,\beta$ & $\geq 0.$ & & & \\
            \hline
        \end{tabular}
        \caption{\label{program: weights} The \textit{weight optimization program} capturing the weight assignment to edges of $\T$ that optimizes the witness complexity of the NBSPwOI and dual adversary vector constructions of Beigi and Taghavi (see~Section~\ref{sec: bt span program}).}
    \end{table}

    The following is our second main theorem.
    \begin{theorem}\label{thm: weights in program give query upper bound}
        Let $f \subseteq \zone^n \times \calR$ be a relation and let $\T$ be a decision tree computing $f$. Then,
        \[
        \Q(\leaff) = O(\OPT_{\T}).
        \]
    \end{theorem}

    Our results give us a new way to bound $\OPT_\T$ and thus $\Q(\leaff)$ from above. Combined with the earlier results from Beigi and Taghavi, this gives us the following corollary. For formal definitions of the measures below, see Section~\ref{sec: prelims}.

    \begin{corollary}
        \label{cor:ubs}
        Let $f \subseteq \zone^n \times \calR$ be a relation and let $\T$ be a deterministic decision tree computing it, weighted with the canonical weight assignment as defined in Definition~\ref{def:weight-assignment}. Then, the quantum query complexity of $\leaff$ (in fact $\OPT_\T$) satisfies the following two bounds:
        \begin{enumerate}
            \item The rank-depth bound: $\Q(\leaff) = O\rbra{\sqrt{\rank(\T)\textnormal{depth}(\T)}}$.
            \item The size bound: $\Q(\leaff) = O\rbra{\sqrt{\dtsize(\T)}}$.
        \end{enumerate}
    \end{corollary}

    Using standard arguments to deal with the case when $\T$ is a randomized decision tree, this gives an upper bound on the bounded-error quantum query complexity of a function in terms of its randomized decision-tree size complexity, as well as in terms of its randomized rank and depth (see Corollary~\ref{cor: quantum upper bound in terms of leaves}).

    It was shown by Reichardt~\cite{Rei11} that the quantum query complexity of evaluating Boolean formulas of size $s$ is $O(\sqrt{s})$. In particular, this implies the size bound in Corollary~\ref{cor:ubs} for Boolean functions $f$ since the formula size of a Boolean function is bounded above by a constant times its decision-tree size (see Claim~\ref{claim: formula size dtsize}). Not only does our bound also hold for \emph{relations}, but the query algorithm we obtain is actually an algorithm for $\leaff$ when the underlying tree is deterministic.
    While this yields trivial bounds for most relations (since almost all Boolean functions have super-polynomial decision-tree size complexity, for example, while $\Q(f)$ is at most $n$), it recovers the $O(\sqrt{n})$ bound for the search problem~\cite{Gro96}, for example. We also exhibit a family of decision trees for which the size bound is strictly stronger than the rank-depth bound (see Figure~\ref{fig:bounds-separation}).

    \subsection{Subsequent work}

    In a preliminary version of this paper~\cite{CMP22} we conjectured that rank and randomized rank are polynomially related for all total Boolean functions. Note that techniques used to prove the analogous polynomial equivalence for deterministic and randomized query complexities~\cite{Nisan91} (also see~\cite{BW02}) cannot work since they use intermediate measures such as certificate complexity, sensitivity and block sensitivity, all of which are maximal for the $\OR$ function even though its rank is 1.

    It is not hard to see that, up to a logarithmic factor in the input size, (randomized) rank is equivalent to logarithm of (randomized) decision tree size complexity.
    Using this equivalence, it was subsequently shown~\cite{CDMRS23} that up to a polylogarithmic factor in the input size our conjecture is true. That is, rank and randomized rank are polynomially related (up to a polylogarithmic factor in the input size) for all total Boolean functions.

    As a consequence, the quantum query upper bounds we obtain in Theorem~\ref{thm: qq upper bound in terms of randomized rank} can only be polynomially better than the upper bound obtained in Theorem~\ref{thm: qq upper bound in terms of rank}.

    An alternative proof for Theorem~\ref{thm: weights in program give query upper bound} also appears in \cite[Section~5.4]{Cor25}.

    \subsection{Organization}
    We provide necessary preliminaries in Section~\ref{sec: prelims}.
    In Section~\ref{sec: dtrank}, we discuss the rank of decision trees, and prove that it is equal to the guessing complexity. We also prove Theorem~\ref{thm: rank rrank separation nand tree}, which is a polynomial separation between rank and randomized rank for the complete binary AND-OR tree, in Section~\ref{sec: dtrank}. In Section~\ref{sec: bt span program} we discuss Beigi and Taghavi's construction of a span program and a dual adversary solution for a relation $f$ by assigning weights to edges of a classical decision tree computing $f$, and we show that the best possible weighting scheme is captured in Program~\ref{program: weights}, proving Theorem~\ref{thm: weights in program give query upper bound}. In Section~\ref{sec: weight assignments}, we exhibit a solution to Program~\ref{program: weights} and argue its optimality from first principles. In Appendix~\ref{app: sizes} we prove an upper bound on formula size of a Boolean function in terms of its decision-tree size. In Appendix~\ref{app: another weight sqrt llogl} we exhibit another interesting, albeit sub-optimal, solution to Program~\ref{program: weights}.

    \section{Preliminaries}\label{sec: prelims}
    All logarithms in this paper are taken base 2. For a bit $b \in \zone$, let $\overline{b}$ denote the bit $1 - b$. For a relation $f \subseteq \zone^n \times \calR$, define the domain of $f$ to be $D_f := \{x \in \zone^n : \exists b \in \calR$ such that $(x, b) \in f\}$. For a vector $v$, let $\norm{v}$ denote its $\ell_2$-norm. For a matrix $M$, let $\norm{M}$ denote its spectral norm.
    For matrices $M$ and $N$ of the same dimensions, let $M \circ N$ denote their entry-wise (Hadamard) product. Let $\delta_{a, b}$ be the function that outputs 1 when $a = b$ and 0 otherwise. We use $[T]$ to denote the set $\{1,\ldots, T\}$ where $T \in \mathbb{Z}^+$.
    A \emph{restriction} is a partial assignment of variables $x_1, \dots, x_n$ to values in $\zone$. Formally, a restriction is a function $\rho :[n] \to \cbra{0, 1, \bot}$, where $\rho(i) = \bot$ indicates that the value of $x_i$ is left free. For a function $f : \zone^n \to \zone$ and a restriction $\rho :[n] \to \cbra{0, 1, \bot}$, the restricted function $f|_\rho$ is defined on the free variables in the natural way.

    \subsection{Decision Trees}

    A decision tree computing a relation $f \subseteq \zone^n \times \calR$ is a binary tree with leaf nodes labeled in $\calR$, each internal node is labeled by a variable $x_i$ and has two outgoing edges, labeled $0$ and $1$. On input $x \in D_f$, the tree’s computation proceeds by computing $x_i$ as indicated by the node’s label and following the edge indicated by the value of the computed variable. The output value at the leaf, say $b \in \calR$, must be such that $(x, b) \in f$.
    Given a relation $f \subseteq \zone^n \times \calR$ and a deterministic decision tree $\T$ computing it, define the function $\leaff$ that takes an input $x \in \zone^n$ and outputs the leaf of $\T$ reached on input $x$. Let $V(\T)$ denote the set of vertices in $\T$, and, we use $V_I(\T)$ to denote the set of internal nodes of $\T$. We use the notation $J(v)$ to denote the variable queried at vertex $v$ for an internal node $v$. For a vertex $v \in V_I(\T)$ and $q \in \zone$, let $N(v, q)$ denote the vertex that is the child of $v$ along the edge that has label $q$. For neighbors $v, w$, let $e(v, w)$ denote the edge between them. For an input $x \in \zone^n$, let $P_x$ denote the unique path in $\T$ from its root to $\leaff(x)$. For a path $P$ in $\T$, we say an edge $e$ \emph{deviates from $P$} if exactly one vertex of $e$ is in $P$. For a path $P$ in $\T$, define $\overline{P} := \cbra{e : e~\textnormal{deviates from}~P}$. We let $P(\T)$ denote the set of all paths from the root to a leaf in $\T$.
    We assume that decision trees computing relations $f$ contain no extraneous leaves, i.e., for all leaves there is an input $x \in D_f$ that reaches that leaf. We also assume that for every path $P$ in $\T$ and index $i \in [n]$, the variable $x_i$ is queried at most once on $P$.

    The decision tree complexity (also called \emph{deterministic query complexity}) of $f$, denoted $\DT(f)$, is defined as follows.
    \[
    \DT(f) := \min_{\T:\T~\textnormal{is a DT computing}~f}
    \textnormal{depth}(\T).
    \]
    Note that a deterministic decision tree in fact computes a function, since each input reaches exactly one leaf on the computation path of the tree.
    A randomized decision tree is a distribution over deterministic decision trees. We say a randomized decision tree computes a relation $f \subseteq \zone^n \times \calR$ with error $1/3$ if for all $x \in \zone^n$, the probability of outputting a $b \in \calR$ such that $(x, b) \in f$ is at least $2/3$. The depth of a randomized decision tree is the maximum depth of a deterministic decision tree in its support.
    Define the randomized decision tree complexity (also called \emph{randomized query complexity}) of $f$ as follows.
    \[
    \RDT(f) := \min_{\substack{\T : \T~\textnormal{is a randomized DT}\\\textnormal{that computes}~f~\textnormal{up to error}~1/3}}
    \textnormal{depth}(\T).
    \]
    Another measure of interest to us in this work is the \emph{decision-tree size complexity} of $f$.

    \begin{definition}[Decision-tree size complexity]\label{defn: leaf complexity}
        Let $f \subseteq \zone^n \times \calR$ be a relation. Define the \emph{decision-tree size complexity} of $f$, which we denote by $\dtsize(f)$, as
        \[
        \dtsize(f) := \min_{\T : \T~\textnormal{computes}~f}\dtsize(\T),
        \]
        where $\dtsize(\T)$ denotes the number of nodes of $\T$. Analogously, the \emph{randomized decision-tree size complexity} of $f$ is defined to be
        \[
        \rdtsize(f):= \min_{\substack{\T : \T~\textnormal{is a randomized DT}\\\textnormal{that computes}~f~\textnormal{up to error}~1/3}}\rdtsize(\T),
        \]
        where $\rdtsize(\T)$ denotes the maximum number of nodes of a decision tree in the support of $\T$.
    \end{definition}
    It is easy to observe that the number of nodes in a deterministic decision tree equals one less than twice the number of leaves in the tree.

    We require the notion of rank of a decision tree introduced by Ehrenfeucht and Haussler~\cite{EH89}. We also require a natural randomized variant of the same, which to the best of our knowledge, has not been studied in the literature.

    \begin{definition}[Decision tree rank and randomized rank]\label{defn: rank}
        Let $\T$ be a binary decision tree. Define the rank of $\T$ recursively as follows:
        For a leaf node $a$, define $\rank(a) = 0$. For an internal node $u$ with children $v, w$, define
        \[
        \rank(u) = \begin{cases}
            \max\cbra{\rank(v), \rank(w)} & \textnormal{if}~\rank(v) \neq \rank(w)\\
            \rank(v) + 1 & \textnormal{if}~\rank(v) = \rank(w).
        \end{cases}
        \]
        Define $\rank(\T)$ to be the rank of the root of $\T$.
        Define the \emph{randomized rank} of a randomized decision tree to be the maximum rank of a deterministic decision tree in its support.
    \end{definition}

    \begin{definition}[Rank of a relation]
        Let $f \subseteq \zone^n \times \calR$ be a relation. Define the rank of $f$, which we denote by $\rank(f)$, by
        \[
        \rank(f) = \min_{\T : \T~\textnormal{computes}~f} \rank(\T).
        \]
        Analogously define the randomized rank of $f$, which we denote by $\rrank(f)$, to be the minimum randomized rank of a randomized decision tree that computes $f$ to error $1/3$.
    \end{definition}

    \subsection{Quantum Query Complexity}\label{sec: prelims quantum}
    We refer the reader to~\cite{NC16, Wol19} for the basics of quantum computing. A quantum query algorithm $\mathcal{A}$ for a relation $f \subseteq \zone^n \times \calR$ begins in a fixed an initial state $\ket{\psi_0}$, applies a sequence of unitaries $U_0, O_x, U_1, O_x, \cdots, U_T$, and performs a measurement. Here, the initial state $\ket{\psi_0}$ and the unitaries $U_0, U_1, \dots, U_T$ are independent of the input. The unitary $O_x$ represents the ``query'' operation, and maps $\ket{i}\ket{b}$ to $\ket{i}\ket{b + x_i \mod 2}$ for all $i \in [n]$ and maps $\ket{i}$ to $\ket{i}$ whenever $i=0$.
    We say that $\mathcal{A}$ is an \emph{$\eps$-error algorithm} computing $f$ if for all $x$ in the domain of $f$, the probability of outputting $b \in \calR$ such that $(x, b) \in f$ is at least $1 - \eps$. The \emph{$\eps$-error quantum query complexity} of $f$, denoted by $\Q_\eps(f)$, is the least number of queries required for a quantum query algorithm to compute $f$ with error $\eps$. When the subscript $\eps$ is dropped we assume $\eps = 1/3$; the \emph{bounded-error query complexity} of $f$ is $\Q(f)$.

    \subsection{Span Programs}

    The model of span programs of interest to us is that of ``non-binary span programs with orthogonal inputs'', abbreviated NBSPwOI. This model was introduced by Beigi and Taghavi~\cite{BT19}.
    A NBSPwOI $(\spanProgramP, w, \overline{w})$ evaluating a function $f:D_f \rightarrow [m]$ with $D_f \subseteq [\ell]^n$ consists of the following:

    \begin{itemize}
        \item A finite dimensional vector space $V$ equipped with an inner product,
        \item for every $i \in [m]$, a target vector $\ket{t_i} \in V$,
        \item for every $j\in [n]$ and every $q \in [\ell]$, an input set $I_{j,q} \subseteq V$.
    \end{itemize}
    Let $I$ be the union of all input sets:
    \begin{equation}
        I=\bigcup_{j \in [n]}\bigcup_{q \in [\ell]} I_{j,q}.
    \end{equation}
    For every $x \in D_f$ we define the set of \emph{available vectors for $x$} to be
    \begin{equation}
        I(x)=\bigcup_{j \in [n]} I_{j,x_j}.
    \end{equation}
    We say that the set of vectors $I \setminus I(x)$ is \emph{unavailable for $x$}.
    The set $P$ comprises all of the above components.
    Let $A$ be the $(d \times |I|)$-dimensional matrix consisting of all input vectors as its columns with $d=\dim(V)$.
    We say $(\spanProgramP, w, \overline{w})$ evaluates $f$ if for every $x \in D_f$,
    \[
    \ket{t_{\alpha}} \in \spann(I(x)) \iff \alpha = f(x).
    \]
    Moreover, there should be two witnesses indicating this: a positive witness $\ket{\posWitX} \in \C^{|I|}$ and a negative witness $\ket{\negWitX} \in V$ satisfying the following conditions:
    \begin{itemize}
        \item The coordinates of $\ket{\posWitX}$ associated to unavailable vectors are zero.
        \item $A\ket{\posWitX}=\ket{t_{\alpha}}$.
        \item For all $\ket{v} \in I(x)$,
        \[
        \braket{v}{\negWitX}=0.
        \]
        \item $\forall \beta \neq \alpha$ we have $\braket{t_{\beta}}{\negWitX}=1$.
    \end{itemize}
    Let $w$ and $\overline{w}$ denote the collections of positive and negative witnesses, respectively. The positive and negative witness sizes of $(\spanProgramP, w, \overline{w})$ are defined to be
    \begin{align}
        \wsizePos(\spanProgramP, w, \overline{w}) & \coloneqq \max_{x \in D_f} \norm{\ket{\posWitX}}^2,\label{eqn: poswit defn} \\
        \wsizeNeg(\spanProgramP, w, \overline{w}) & \coloneqq \max_{x \in D_f} \norm{A^{\dagger}\ket{\negWitX}}^2. \label{eqn: negwit defn}
    \end{align}
    The \emph{witness size} of $(\spanProgramP, w, \overline{w})$ is defined as
    \begin{equation}\label{eqn: defn wsize}
        \wsize(\spanProgramP, w, \overline{w}) \coloneqq \sqrt{\wsizeNeg(\spanProgramP, w, \overline{w}) \cdot \wsizePos(\spanProgramP, w, \overline{w})}.
    \end{equation}
    Beigi and Taghavi~\cite[Theorem~2]{BT19} showed that for any NBSPwOI $(\spanProgramP, w, \overline{w})$ evaluating the function $f$, its
    complexity $\wsize(\spanProgramP, w,\overline{w})$ is an upper bound on $\Q(f)$.

    \begin{theorem}[\cite{BT19}]\label{thm: bt19}
        Let $f : D_f \to [m]$ be a function with $D_f \subseteq [\ell]^n$, and let $(P, w, \overline{w})$ be a NBSPwOI computing $f$. Then,
        \[
        \Q(f) = O(\wsize(\spanProgramP, w, \overline{w})).
        \]
    \end{theorem}

    \subsection{Dual Adversary Bound}
    \label{sec:DualAdversaryBound}
    Let $f:D_f \rightarrow [m]$ be a function with $D_f \subseteq [\ell]^n$. Let $\Gamma$ be a Hermitian matrix with rows and columns labeled by elements of $D_f$. The matrix $\Gamma$ is called an \emph{adversary matrix} for $f$ if $\Gamma[x,y]=0$ whenever $f(x)=f(y)$. H{\o}yer, Lee and \v{S}palek~\cite{HLS07} defined the quantity
    \begin{equation}
        \label{eq:GeneralisedAdversarySolution}
        \Adv^{\pm}(f)= \max_{\Gamma \neq 0} \frac{\norm{\Gamma}}{\max_{i} \norm{\Gamma \circ D_i}},
    \end{equation}
    where $D_i$ is a $\zone$-valued matrix with $D_i[x, y] = 1 - \delta_{x_i, y_i}$. They showed that $\Q(f)=\Omega(\Adv^{\pm}(f))$. Reichardt~\cite{Rei11} subsequently proved that $\Q(f)=\Theta(\Adv^{\pm}(f))$ for Boolean functions $f:\zone^n \rightarrow \ZO$. This result was later generalized by Lee et al.~\cite[Theorem~1.1]{LMRSS11} who showed
    \begin{equation}
        \Q(f)=\Theta(\Adv^{\pm}(f))
    \end{equation} for non-Boolean functions as well. They first observed that the quantity $\Adv^{\pm}(f)$ in Equation~\eqref{eq:GeneralisedAdversarySolution} can be viewed as a feasible solution to a semidefinite program (SDP),
    \begin{equation}
        \label{eq:SDPprimal}
        \Adv^{\pm}(f) = \max \{\norm{\Gamma} : \forall i \in [n], \norm{\Gamma \circ D_i} \le 1\},
    \end{equation}
    where the maximization is over all real adversary matrices for $f$, i.e., symmetric matrices $\Gamma$ such that $\Gamma[x,y]=0$ whenever $f(x)=f(y)$. While the primal solution of SDP in Equation~\eqref{eq:SDPprimal} gives a lower bound on the quantum query complexity of $f$, they argued that a solution to the dual of this SDP gives upper bounds on $\Q(f)$. Based on duality of SDPs, they first showed (\cite[Lemma~A.1]{LMRSS11}) that the dual SDP is of the following form.
    \begin{table}[H]
        \centering~
        \begin{tabular}{|lll|}
            \hline
            Variables & $\cbra{\ket{u_{xj}} : x \in D_f, j \in [n]}$ and $\cbra{\ket{w_{xj}} : x \in D_f, j \in [n]}, d$ & \\
            Minimize  & $\max_{x \in D_f} \max\cbra{\sum_{j=1}^{n} \norm{\ket{u_{xj}}}^2, \sum_{j=1}^{n}\norm{\ket{w_{xj}}}^2} $ & \\
            s.t.
            & $\sum_{j \in [n] : x_j\neq y_j} \braket{u_{xj}}{w_{yj}} = 1- \delta_{f(x),f(y)}$ & $\forall x,y \in D_f$ \\
            & $\ket{u_{xj}}, \ket{w_{xj}} \in \mathbb{C}^{d}$ & for all $x \in D_f$ \\
            \hline
        \end{tabular}
        \caption{\label{program: sdp dual f} Dual SDP for $f$}
    \end{table}

    They then showed that a feasible solution to the SDP in Program~\ref{program: sdp dual f} yields an upper bound on the quantum query complexity of $f$ (\cite[Theorem~4.1]{LMRSS11}).
    \begin{theorem}[\cite{LMRSS11}]
        \label{thm:DualSDP-QuantumQueryUpperBound}Let $f:D_f \rightarrow [m]$ be a function with $D_f \subseteq [\ell]^n$, let $C$ denote the optimal value of Program~\ref{program: sdp dual f} for $f$. Then
        \begin{equation}
            \Q(f)=O(C).
        \end{equation}
    \end{theorem}

    \section{Decision Tree Rank}\label{sec: dtrank}

    In this section, we first rephrase Theorem~\ref{thm: linlin} in terms of a measure of decision trees which we term `guessing complexity'. This reformulation was essentially done by Beigi and Taghavi~\cite[Section 3]{BT20}. We then show that the guessing complexity of a decision tree equals its rank, proving Theorem~\ref{thm: qq upper bound in terms of rank}. Finally, we show a polynomial separation between rank and randomized rank for the complete binary AND-OR tree.

    \subsection{Guessing Complexity and Rank}

    \begin{definition}[G-coloring~{\cite[Definition 1]{BT20}}]
        A G-coloring of a decision tree $\T$ is a coloring of its edges by two colors black and red, in such a way that any vertex of $\T$ has at most one outgoing edge with black color.
    \end{definition}

    \begin{definition}[Guessing Complexity]
        \label{def:guessing-complexity}
        Let $\T$ be a decision tree. Let $P(\T)$ denote the set of root-to-leaf paths in $\T$. Define the \emph{guessing complexity} of $\T$, which we denote by $G(\T)$, by
        \[
        G(\T) = \min_{\textnormal{G-colorings of~}\T}\max_{P \in P(\T)}\textnormal{number of red edges on $P$}.
        \]
    \end{definition}

    \begin{claim}\label{claim: gcoloring equals rank}
        Let $\T$ be a decision tree. Then,
        \[
        G(\T) = \rank(\T).
        \]
    \end{claim}
    \begin{proof}
        Let $v$, $v_L$ and $v_R$ be the root of $\T$, and the left and right children of $v$, respectively. Let $\T_L$ and $\T_R$ denote the subtrees of $\T$ rooted at $v_L$ and $v_R$, respectively.

        Consider a $G$-coloring of $\T$. This naturally induces a $G$-coloring of $\T_L$ and $\T_R$.
        We consider two cases:
        \begin{itemize}
            \item $G(\T_L) = G(\T_R) = k$, say. One of the edges $(v, v_L)$ or $(v, v_R)$ must be colored red. Assume without loss of generality that $(v, v_L)$ is the red edge. Since we assumed $G(\T_L) = k$, $\T_L$ contains a path with at least $k$ red edges under the $G$-coloring induced from the given $G$-coloring of $\T$. But this induces a path in $\T$ with $k+1$ red edges, and hence $G(\T) \geq G(\T_L) + 1$.
            \item If $G(\T_L) \neq G(\T_R)$, we have $G(\T) \geq \max\cbra{G(\T_L),G(\T_R)}$, witnessed by the $G$-colorings induced on $\T_L$ and $\T_R$ by the $G$-coloring of $\T$.
        \end{itemize}
        In the other direction, we construct an optimal $G$-coloring of $\T$ given optimal $G$-colorings of $\T_L$ and $\T_R$. The edges of $\T_L$ and $\T_R$ in $\T$ are colored exactly as they are in the given optimal $G$-colorings of them. It remains to assign colors to the two remaining edges $(v, v_L)$ and $(v, v_R)$. We again have two cases:
        \begin{itemize}
            \item $G(\T_L) = G(\T_R) = k$, say. Arbitrarily color one of the edges $(v, v_L)$ and $(v, v_R)$ (say, $(v, v_L)$) red, and color the other edge black. The maximum number of red edges on a path has increased by 1. Thus, $G(\T) \leq G(\T_L) + 1$.
            \item $G(\T_L) > G(\T_R)$, say (the other case follows a similar argument). Color the edge $(v, v_L)$ black and $(v, v_R)$ red. Thus, the maximum number of red edges on a path in $\T$ equals $\max\cbra{G(\T_L), G(\T_R) + 1} = G(\T_L)$.
        \end{itemize}

        Thus, we have
        \[
        G(\T) = \begin{cases}
            \max\{G(\T_L),G(\T_R)\} & \textnormal{if}~G(\T_L)\neq G(\T_R) \\
            G(\T_L)+1 & \textnormal{if}~G(\T_L)=G(\T_R),
        \end{cases}
        \]
        The measure $\rank(\T)$ is defined exactly as the above (Definition~\ref{defn: rank}), proving the claim.
    \end{proof}

    The guessing algorithm $\G$ in Theorem~\ref{thm: linlin} corresponds to a natural G-coloring of $\T$ of cost $G$: for each internal vertex, color the guessed edge black and the other edge red. Thus, Theorem~\ref{thm: qq upper bound in terms of rank} immediately follows from Claim~\ref{claim: gcoloring equals rank} and Theorem~\ref{thm: linlin}.

    \begin{proof}[Proof of Theorem~\ref{thm: qq upper bound in terms of randomized rank}]
        A quantum query algorithm for $f$ is as follows: sample a decision tree from the support of $\T$ according to its underlying distribution, and run a $9/10$-error quantum query algorithm of cost $O(\sqrt{TG})$ from Theorem~\ref{thm: qq upper bound in terms of rank} on the resultant tree.\footnote{This is possible since the deterministic decision trees in the support of $\T$ compute \emph{functions}, which admit efficient error reduction with a constant overhead in query complexity by standard techniques (run the algorithm from Theorem~\ref{thm: qq upper bound in terms of rank} a large constant many times and return the majority output).} The success probability of this algorithm is at least $(9/10)\cdot(2/3) = 3/5$ for all inputs $x \in D_f$.
    \end{proof}

    The rank of a decision tree essentially captures the largest depth of a complete binary subtree of the original tree. Thus, the rank of a tree is bounded from above by the logarithm of the size of the tree.
    \begin{observation}[{\cite[Lemma 1]{EH89}}]
        Let $\T$ be a deterministic decision tree of size $s$. Then,
        \[
        \rank(\T) \leq \log (s+1)-1.
        \]
    \end{observation}
    Along with this observation and the simple observation that the depth of a decision tree is at most its size, Theorem~\ref{thm: qq upper bound in terms of randomized rank} yields the following statement.
    \begin{theorem}\label{thm: llogl query upper bound}
        Let $\T$ be a randomized decision tree of size $s$ that computes a relation $f \subseteq \zone^n \times \calR$. Then,
        \[
        \Q_{2/5}(f) = O(\sqrt{s \log s}).
        \]
    \end{theorem}

    Note here that it suffices to prove that $\Q(\leaff) = O(\sqrt{s \log s})$ where $s$ is the size of a \emph{deterministic} decision tree computing $f$, since standard techniques and error reduction yield the required bound for randomized trees.
    In Appendix~\ref{app: another weight sqrt llogl} we show an explicit NBSPwOI and dual adversary solution witnessing the same bound.
    In Sections~\ref{sec: bt span program} and~\ref{sec: weight assignments} we show an explicit NBSPwOI and dual adversary solution witnessing a stronger bound without the logarithmic factor. We choose to still give the weaker bound in Appendix~\ref{app: another weight sqrt llogl} as the weighting scheme seems considerably different from that in Section~\ref{sec: weight assignments}, and the weights are also efficiently computable.

    We note here the equivalence of the rank of a Boolean function and the value of an associated Prover-Delayer game introduced by Pudlák and Impagliazzo~\cite{PI00}. We use this equivalence in the next part of this section to show that the rank of the complete binary AND-OR tree is polynomially larger than its randomized rank.

    The game is played between two players: the Prover and the Delayer, who construct a partial assignment, say $\rho \in \cbra{0, 1, \bot}^n$, in rounds. To begin with the assignment is empty, i.e., $\rho = \bot^n$. In a round, the Prover queries an index $i \in [n]$ for which the value $x_i$ is not set in $\rho$ (i.e., $\rho_i = \bot$). The Delayer either answers $x_i = 0$ or $x_i = 1$, or defers the choice to the Prover. In the latter case, the Delayer scores a point. The game ends when the Prover knows the value of the function, i.e., when $f|_\rho$ is a constant function. The value of the game, $\val(f)$, is the maximum number of points the Delayer can score regardless of the Prover's strategy.
    The following result is implicit in~\cite{PI00} (also see~\cite[Theorem~3.1]{DM21} for an explicit statement and proof).
    \begin{claim}\label{claim: rank equals Prover Delayer game value}
        Let $f : \zone^n \to \zone$ be a (possibly partial) Boolean function. Then,
        \[
        \rank(f) = \val(f).
        \]
    \end{claim}

    \subsection{A Separation Between Rank and Randomized Rank}
    We first note that there can be maximal separations between rank and randomized rank if we consider partial functions. This is witnessed by the well-studied Approximate-Majority function, for example.
    \begin{claim}\label{claim: apxmaj rank rrank}
        Let $f : \mathcal{D} \subset \zone^n \to \zone$ be a partial function defined as follows:
        \[
        f(x) = \begin{cases}
            0 & |x| \leq n/3\\
            1 & |x| \geq 2n/3.
        \end{cases}
        \]
        Then, $\rank(f) = \Theta(n)$ and $\rrank(f) = \Theta(1)$.
    \end{claim}
    \begin{proof}
        Clearly $\rank(f) = O(n)$. For the lower bound, we use the equivalence from Claim~\ref{claim: rank equals Prover Delayer game value}. A valid Delayer strategy is as follows: allow the Prover to choose input values for their first $n/3$ queries. It is easy to see that no matter what values the Prover chooses, the function can never be restricted to become a constant after these $n/3$ queries. Thus, $\val(f) \geq n/3$ (and this can be easily seen to be tight). The randomized rank upper bound follows from the easy fact that $\RDT(f) = O(1)$ and $\rrank(f) \leq \RDT(f)$ for all $f$.
    \end{proof}

    When we restrict $f$ to be a total function, it is no longer clear whether or not randomized rank can be significantly smaller than rank. In view of the example above, one might be tempted to consider functions that witness maximal separations between deterministic and randomized query complexity. The current state-of-the-art separation of $\DT(f) = \Omega(n)$ vs.~$\RDT(f) = \tilde{O}(\sqrt{n})$ is witnessed by variants of `pointer jumping' functions~\cite{GPW18, ABBLSS17, MRS18}.
    One might hope to use a similar argument as in the proof of Claim~\ref{claim: apxmaj rank rrank} to show that $\rank(f) = \Omega(n)$ (and a randomized rank upper bound of $\tilde{O}(\sqrt{n})$ immediately follows from the randomized query upper bound). However, it is not hard to show that the rank of these functions is actually $\tilde{O}(\sqrt{n})$, rendering this approach useless for these variants of pointer jumping functions.
    Nevertheless, we are able to use such an approach to show a separation between rank and randomized rank for another function whose deterministic and randomized query complexities are polynomially separated.
    After an initial version of this paper was made public, a subsequent work~\cite{CDMRS23} used the equivalence between (randomized) rank and logarithm of (randomized) decision tree size to show a polynomial equivalence (up to a polylogarithmic factor in the input size) between rank and randomized rank for all Boolean functions. They also observed that a deterministic-randomized query complexity separation can be lifted to the same deterministic-randomized rank separation for a related function~\cite[Appendix~C]{CDMRS23}. Thus, the pointer jumping functions~\cite{GPW18, ABBLSS17, MRS18} can be used to construct a function witnessing a quadratic separation between rank and randomized rank. Finding the best possible (polynomial) separation between rank and randomized rank remains an interesting open question.

    In the remainder of this section, let $F : \zone^n \to \zone$ be defined as the function evaluated by a complete $(\log n)$-depth binary tree as described below. Assume $\log n$ to be an even integer, and the top node of this tree to be an $\OR$ gate. Nodes in subsequent layers alternate between $\AND$'s and $\OR$'s, and nodes at the bottom layer contain variables. We call this the complete AND-OR tree of depth $\log n$. See Figure~\ref{fig: andor tree} for a depiction of the complete depth-4 AND-OR tree.

    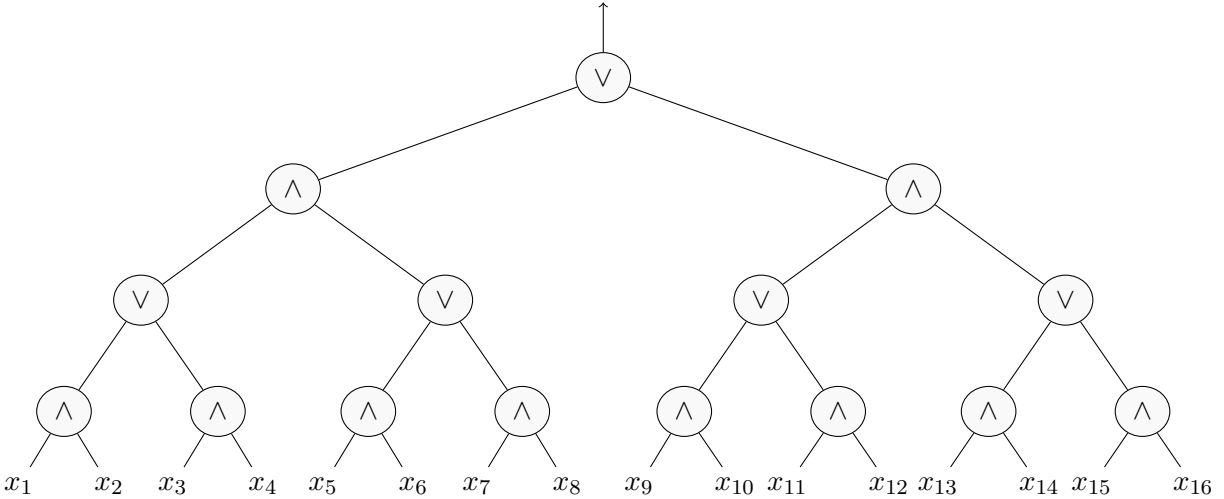
\begin{figure}
        \begin{center}
            \begin{tikzpicture}
                \node[comm] (v)   [] {$\vee$};
                \coordinate[above of = v] (output);
                \node[comm] (v0)  [below left=1 and 3 of v] {$\wedge$};
                \node[comm] (v1)  [below right=1 and 3 of v] {$\wedge$};

                \node[comm] (v00) [below left=1 and 1.3 of v0] {$\vee$};
                \node[comm] (v01) [below right=1 and 1.3 of v0] {$\vee$};
                \node[comm] (v10) [below left=1 and 1.3 of v1] {$\vee$};
                \node[comm] (v11) [below right=1 and 1.3 of v1] {$\vee$};

                \node[comm] (v000)  [below left=1 and .5 of v00] {$\wedge$};
                \node[comm] (v001)  [below right=1 and .5 of v00] {$\wedge$};
                \node[comm] (v010)  [below left=1 and .5 of v01] {$\wedge$};
                \node[comm] (v011)  [below right=1 and .5 of v01] {$\wedge$};
                \node[comm] (v100)  [below left=1 and .5 of v10] {$\wedge$};
                \node[comm] (v101)  [below right=1 and .5 of v10] {$\wedge$};
                \node[comm] (v110)  [below left=1 and .5 of v11] {$\wedge$};
                \node[comm] (v111)  [below right=1 and .5 of v11] {$\wedge$};

                \node[leaf] (v0000)  [below left=.5 and -.15 of v000] {$x_1$};
                \node[leaf] (v0001)  [below right=.5 and -.15 of v000] {$x_2$};
                \node[leaf] (v0010)  [below left=.5 and -.15 of v001] {$x_3$};
                \node[leaf] (v0011)  [below right=.5 and -.15 of v001] {$x_4$};
                \node[leaf] (v0100)  [below left=.5 and -.15 of v010] {$x_5$};
                \node[leaf] (v0101)  [below right=.5 and -.15 of v010] {$x_6$};
                \node[leaf] (v0110)  [below left=.5 and -.15 of v011] {$x_7$};
                \node[leaf] (v0111)  [below right=.5 and -.15 of v011] {$x_8$};
                \node[leaf] (v1000)  [below left=.5 and -.15 of v100] {$x_9$};
                \node[leaf] (v1001)  [below right=.5 and -.15 of v100] {$x_{10}$};
                \node[leaf] (v1010)  [below left=.5 and -.15 of v101] {$x_{11}$};
                \node[leaf] (v1011)  [below right=.5 and -.15 of v101] {$x_{12}$};
                \node[leaf] (v1100)  [below left=.5 and -.15 of v110] {$x_{13}$};
                \node[leaf] (v1101)  [below right=.5 and -.15 of v110] {$x_{14}$};
                \node[leaf] (v1110)  [below left=.5 and -.15 of v111] {$x_{15}$};
                \node[leaf] (v1111)  [below right=.5 and -.15 of v111] {$x_{16}$};

                \draw [-] (v)   --  (v0);
                \draw [-] (v)   --  (v1);
                \draw [-] (v0)   --  (v00);
                \draw [-] (v0)   --  (v01);
                \draw [-] (v1)   --  (v10);
                \draw [-] (v1)   --  (v11);
                \draw [-] (v00)   --  (v000);
                \draw [-] (v00)   --  (v001);
                \draw [-] (v01)   --  (v010);
                \draw [-] (v01)   --  (v011);
                \draw [-] (v10)   --  (v100);
                \draw [-] (v10)   --  (v101);
                \draw [-] (v11)   --  (v110);
                \draw [-] (v11)   --  (v111);

                \draw [-] (v000)   --  (v0000);
                \draw [-] (v000)   --  (v0001);
                \draw [-] (v001)   --  (v0010);
                \draw [-] (v001)   --  (v0011);
                \draw [-] (v010)   --  (v0100);
                \draw [-] (v010)   --  (v0101);
                \draw [-] (v011)   --  (v0110);
                \draw [-] (v011)   --  (v0111);
                \draw [-] (v100)   --  (v1000);
                \draw [-] (v100)   --  (v1001);
                \draw [-] (v101)   --  (v1010);
                \draw [-] (v101)   --  (v1011);
                \draw [-] (v110)   --  (v1100);
                \draw [-] (v110)   --  (v1101);
                \draw [-] (v111)   --  (v1110);
                \draw [-] (v111)   --  (v1111);

                \draw [->] (v)    --  (output);
            \end{tikzpicture}
            \caption{Complete AND-OR tree of depth 4}\label{fig: andor tree}
        \end{center}
    \end{figure}

    It is easy to see via an adversarial argument that $\DT(F) = n$. Saks and Wigderson~\cite{SW86} showed that $\RDT(F) = \Theta(n^{\log\frac{1 + \sqrt{33}}{4}}) \approx \Theta(n^{0.753\dots})$.

    \begin{theorem}[{\cite[Theorem~1.5]{SW86}}]\label{thm: saks wigderson}
        Let $F : \zone^n \to \zone$ be the complete AND-OR tree of depth $\log n$. Then
        \[
        \DT(F) = n, \qquad \RDT(F) = \Theta\rbra{n^{\log\frac{1 + \sqrt{33}}{4}}} \approx \Theta(n^{0.753\dots}).
        \]
    \end{theorem}

    We show that the rank of $F$ equals $(n + 2)/3$.
    \begin{theorem}\label{thm: rank of complete nand tree}
        Let $F : \zone^n \to \zone$ be the complete AND-OR tree of depth $\log n$. Then
        \[
        \rank(F) = \frac{n + 2}{3}.
        \]
    \end{theorem}

    This proof uses the characterization of rank by the value of the Prover-Delayer game in Claim~\ref{claim: rank equals Prover Delayer game value}. For the lower bound, we exhibit a strategy of the Delayer that scores $(n + 2)/3$ points regardless of the strategy of the Prover. For the upper bound, we show a strategy of the Prover which ensures that the Delayer can score at most $(n + 2)/3$ points.

    Before we exhibit these strategies, we describe how partial assignments to the AND-OR tree can be captured by modifications of the tree itself. To that end, we formalize required properties of AND-OR trees in the following definition.

    \begin{definition}
        Let $n \in \N$, and let $\T$ be a tree, whose internal nodes are labeled by either $\land$ or $\lor$, and whose leaves are labeled by indices in $[n]$, such that no two leaves are labeled by the same index. Then, we refer to $\T$ as an \emph{AND-OR tree}. Edges are directed from the root to leaves, and if $\T$ contains no nodes, we say that it is empty. When $\T$ is non-empty, then the corresponding Boolean function it computes is denoted by $f_{\T} : \{0,1\}^S \to \{0,1\}$, where $S \subseteq [n]$ is the set of index labels that correspond to the leaves present in $\T$. If $\T$ is empty, we say that it computes a constant function, and we have to additionally supply the outcome $b \in \{0,1\}$. We often identify $\T$ with $f_\T$.

        We define some notation to describe the nodes within $\T$:
        \begin{enumerate}
            \item We denote the root node of $\T$ by $\root(\T)$.
            \item For a leaf node $v$, we denote the variable it is querying by $\index(v) \in [n]$.
            \item For a node $v$, define $\llabel(v) \in \{\land, \lor, \leaf\}$ to be the type of node.
            \item For a non-root node $v$, define $\parent(v)$ to be the node in the tree that is the parent of $v$.
            \item For an internal node $v$, define $\child(v)$ to be the set
            \[\child(v)\coloneqq \{ w \in \T: v = \parent(w)\}.\]
            \item For a non-root node $v$, we define its \emph{proper parent}, $\pparent(v)$, recursively as
            \[\pparent(v) = \begin{cases}
                \mathsf{undefined}, & \text{if } v = \root(\T), \\
                \parent(v), & \text{if } |\child(\parent(v))| \geq 2, \\
                \pparent(\parent(v)), & \text{otherwise}.
            \end{cases}\]
        \end{enumerate}
        We say that $\T$ is in \emph{reduced form} if all internal nodes have at least two children, and none of these are labeled by the same gate as the parent.
    \end{definition}

    By fixing a subset of the input variables as per a partial assignment $\rho \in \{0,1,\bot\}^n$, any AND-OR tree can be simplified to a reduced tree $\T|_{\rho}$, i.e., leaves can be removed, and in many cases gates can be removed as well. When we fix a bit in our partial assignment, then the resulting simplification process can be split into two steps, which we refer to as the $\update$ and the $\contract$ steps, respectively. We formally define the corresponding routines in Algorithms~\ref{alg: update}~and~\ref{alg: contract}.

    \begin{algorithm}[H]
        \caption{$\update$ subroutine}\label{alg: update}
        \begin{algorithmic}[1]
            \Require $\T, i \in [n], b \in \{0,1\}$.
            \If{$\exists$ leaf $\ell$ in $\T$ such that $\index(\ell) = i$}
            \If{$\ell$ is the only leaf in $\T$}
            \State $\T \gets $ the empty tree computing $b$.
            \ElsIf{($b = 1$ and $\llabel(\parent(\ell)) = \wedge$) or ($b = 0$ and $\llabel(\parent(\ell)) = \vee$)}
            \State Remove the leaf $\ell$ and the edge above it from $\T$.
            \ElsIf{$\parent(\ell) = \root(\T)$}
            \State $\T \gets $ the empty tree computing $b$.
            \Else
            \State Remove the edge above $\parent(\ell)$, and the subtree rooted at $\parent(\ell)$.
            \EndIf
            \EndIf
            \Ensure $\T$.
        \end{algorithmic}
    \end{algorithm}

    \begin{algorithm}[H]
        \caption{$\contract$ subroutine}\label{alg: contract}
        \begin{algorithmic}[1]
            \Require $\T$
            \While{there is a node $v$ in $\T$ with $|\child(v)|= 1$}
            \If{$v$ is not $\root(\T)$}
            \State $\parent(\child(v)) \gets \parent(v)$
            \EndIf
            \State Remove $v$ from $\T$.
            \EndWhile
            \While{there are vertices $a, a' \in \T$ with $a' = \parent(a)$ and $\llabel(a) = \llabel(a')$}
            \For{$w \in \child(a)$}
            \State $\parent(w) \gets a'$
            \EndFor
            \State Remove $a$ from $\T$.
            \EndWhile
            \Ensure $\T$.
        \end{algorithmic}
    \end{algorithm}

    We now formally prove that the $\update$ procedure modifies the AND-OR tree in such a way that the Boolean function computed by the updated tree indeed computes the function computed by the original tree after updating the partial assignment accordingly. This is the objective of the following claim.

    \begin{claim}
        Let $f : \{0,1\}^n \to \{0,1\}$ be a Boolean function, and let $\rho \in \{0,1,\bot\}^n$ be a partial assignment. Let $\T$ be an AND-OR tree in reduced form, computing $f|_{\rho}$. Let $i \in [n]$ be such that $\rho_i = \bot$, and let $b \in \{0,1\}$. Then, the AND-OR tree $\T' = \update(\T,i,b)$ computes the Boolean function $f|_{\rho'}$, where $\rho'_i = b$, and $\rho'_j = \rho_j$ for $j \neq i$.
    \end{claim}
    \begin{proof}
        We consider multiple cases.
        \begin{itemize}
            \item	If there does not exist a leaf $\ell$ in $\T$ such that $\index(\ell) = i$, then $f|_{\rho}$ does not depend on the $i$'th input bit. Therefore, $f|_{\rho} = f|_{\rho'}$. Hence $\T'$, which is the same as $\T$, computes $f_{\rho'}$.

            \item Suppose such a leaf does exist and let $\ell$ be the leaf in $\T$ such that $\index(\ell) = i$. If $\ell$ is the only leaf present in $\T$, then $f|_{\rho}(x) = x_i$, which means that after setting $\rho_i' = b$, the function value is determined, i.e., $f|_{\rho'} = b$, which corresponds to the tree $\T'$ being empty.

            \item If $\ell$ is a child of an $\land$-node, say $v$, and $b = 1$, then the output of $v$ is completely determined by the other edges going into it. Since we assumed $\T$ to be in reduced form, at least one other such edge must exist. Therefore, the node $\ell$ can be removed, and the new tree $\T'$ computes $f|_{\rho'}$. The same argument holds when $v$ is an $\lor$-node and $b = 0$.

            \item Now suppose that $\parent(\ell)$ is the root of the tree, it is labeled by $\land$, and $b = 0$. In that case, we the $\land$ evaluated to $0$, and therefore the function value is completely determined. Thus, $f|_{\rho'} = 0$, and we output the empty tree computing $0$. The same argument applies when $\parent(\ell)$ is the root of the tree, it is labeled by $\lor$, and $b = 1$.

            \item Finally, if $\ell$ is a child of an $\land$-node that is not the root of $\T$, say $v$, and $b = 0$, then the output of $v$ is always $0$, and the other inputs to $v$ are now redundant. Therefore, we can remove $v$ and the subtree rooted at it. The same holds for the case when $v$ is an $\lor$-node and $b = 1$.
        \end{itemize}

        This completes the proof.
    \end{proof}

    Note that the tree that $\update$ outputs might not be in reduced form. Next, we show that the $\contract$ routine takes as input an AND-OR tree, converts it into reduced form, and does not change the function it computes.

    \begin{claim}
        Let $f : \{0,1\}^n \to \{0,1\}$ be a Boolean function, and let $\rho \in \{0,1,\bot\}^n$ be a partial assignment. Let $\T$ be an AND-OR tree computing $f|_{\rho}$. Then $\T' = \contract(\T)$ is an AND-OR tree in reduced form that computes $f|_{\rho}$.
    \end{claim}

    \begin{proof}
        We analyze the two \texttt{while} loops in Algorithm~\ref{alg: contract}.
        \begin{itemize}
            \item
            In the loop beginning at Line 1, we are selecting internal nodes that have exactly one child. Since these internal nodes are either $\land$- or $\lor$-gates, they act as the identity function, which means that they are redundant. Therefore, they can be removed without altering the function that is being computed, and the their child can be directly connected to their parent, if the node itself is not the root node to begin with. This is precisely what is done in Lines 2-4. After this first \texttt{while} block, we ensure that all internal nodes have at least two children, and the function computed by the tree has not changed. This part of the algorithm is described in Figure~\ref{fig:step1}.

            \begin{figure}[h!]
                \centering
                \begin{tabular}{ccc}
                    \begin{tikzpicture}[gate/.style={draw,circle,minimum size=1.2em,inner sep=0}]
                        \node[gate] (child) at (0,0) {};
                        \node[above right=.3em] at (child) {$w$};
                        \node[gate] (parent) at (0,1) {};
                        \node[above right=.3em] at (parent) {$v$};
                        \node (grandparent) at (0,2) {$\triangle$};
                        \node (leftchild) at (-1,-1) {$\triangle$};
                        \node at (0,-1) {$\cdots$};
                        \node (rightchild) at (1,-1) {$\triangle$};
                        \draw[->] (grandparent) to (parent);
                        \draw[->] (parent) to (child);
                        \draw[->] (child) to (leftchild);
                        \draw[->] (child) to (rightchild);
                    \end{tikzpicture} & \raisebox{18mm}{$\mapsto$} & \begin{tikzpicture}[gate/.style={draw,circle,minimum size=1.2em,inner sep=0}]
                        \node[gate] (new) at (0,.5) {};
                        \node[above right=.3em] at (new) {$w$};
                        \node (grandparent) at (0,2) {$\triangle$};
                        \node (leftchild) at (-1,-1) {$\triangle$};
                        \node at (0,-1) {$\cdots$};
                        \node (rightchild) at (1,-1) {$\triangle$};
                        \draw[->] (grandparent) to (new);
                        \draw[->] (new) to (leftchild);
                        \draw[->] (new) to (rightchild);
                    \end{tikzpicture} \\
                \end{tabular}
                \caption{Graphical depiction of the first step in the $\contract$ procedure where $|\child(v)| = 1$}
                \label{fig:step1}
            \end{figure}
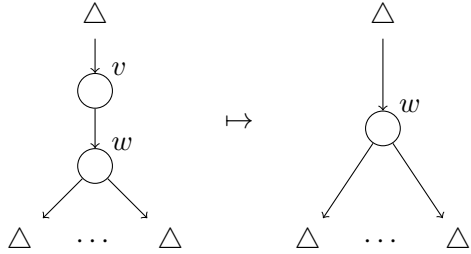

            \item In the loop beginning at Line 5, we select pairs of neighboring vertices (say $v$ and $w = \parent(v)$) that have the same label. In this case, $v$ and $w$ can be contracted without changing the function being computed (due to the associativity of $\land$ and $\lor$). Concretely, this means that the algorithm takes the children of $v$, and adds them as children to $w$ (Line 6-7), after which $v$ is removed (Line 8). This can only increase the number of children of $w$, which means that its number of children must still be at least two. Moreover, this process is repeated until there are no more neighboring vertices that have the same label. Therefore, the resulting tree $\T'$ is in reduced form. This step is displayed in Figure~\ref{fig:step2}.

            \begin{figure}[h!]
                \centering
                \begin{tabular}{ccc}
                    \begin{tikzpicture}[gate/.style={draw,circle,minimum size=1.2em,inner sep=0}]
                        \node[gate] (child) at (0,0) {$\cdot$};
                        \node[above right=.3em] at (child) {$v$};
                        \node[gate] (parent) at (-1.5,1) {$\cdot$};
                        \node[above right=.3em] at (parent) {$w$};
                        \node (leftmiddlechild) at (-3,0) {$\triangle$};
                        \node at (-2,0) {$\cdots$};
                        \node (rightmiddlechild) at (-1,0) {$\triangle$};
                        \node (grandparent) at (-1.5,2) {$\triangle$};
                        \node (leftchild) at (-1,-1) {$\triangle$};
                        \node at (0,-1) {$\cdots$};
                        \node (rightchild) at (1,-1) {$\triangle$};
                        \draw[->] (grandparent) to (parent);
                        \draw[->] (parent) to (child);
                        \draw[->] (parent) to (leftmiddlechild);
                        \draw[->] (parent) to (rightmiddlechild);
                        \draw[->] (child) to (leftchild);
                        \draw[->] (child) to (rightchild);
                    \end{tikzpicture} & \raisebox{18mm}{$\mapsto$} & \begin{tikzpicture}[gate/.style={draw,circle,minimum size=1.2em,inner sep=0}]
                        \node[gate] (parent) at (-1.5,1) {$\cdot$};
                        \node[above right=.3em] at (parent) {$w$};
                        \node (leftmiddlechild) at (-3,0) {$\triangle$};
                        \node at (-2,0) {$\cdots$};
                        \node (rightmiddlechild) at (-1,0) {$\triangle$};
                        \node (grandparent) at (-1.5,2) {$\triangle$};
                        \node (leftchild) at (-1,-1) {$\triangle$};
                        \node at (0,-1) {$\cdots$};
                        \node (rightchild) at (1,-1) {$\triangle$};
                        \draw[->] (grandparent) to (parent);
                        \draw[->] (parent) to (leftmiddlechild);
                        \draw[->] (parent) to (rightmiddlechild);
                        \draw[->] (parent) to (leftchild);
                        \draw[->] (parent) to (rightchild);
                    \end{tikzpicture} \\
                \end{tabular}
                \caption{Graphical depiction of the second step in the $\contract$ procedure.}
                \label{fig:step2}
            \end{figure}
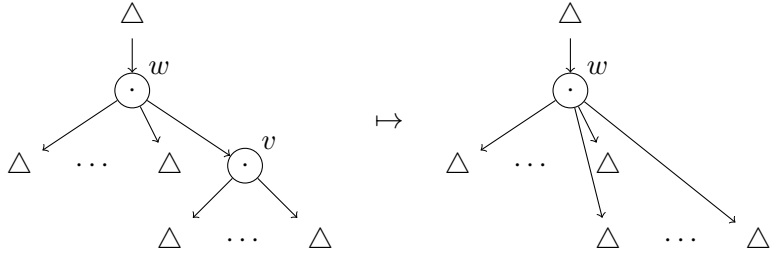
        \end{itemize}
        This completes the proof.
    \end{proof}

    Now that we have described two routines to modify an AND-OR tree $\T$, we can describe the specific strategies for the Delayer and Prover, that witness lower and upper bounds, respectively, on the rank of the complete binary AND-OR tree.

    \subsubsection{Delayer's Strategy (Lower Bound)} The Delayer starts by constructing the complete binary AND-OR tree $\T$ on $n$ bits. Next, whenever the Prover asks variable $i$, Algorithm~\ref{alg: Delayer strat} describes the procedure that the Delayer follows. The output of this algorithm is a new AND-OR tree $\T$, which is being used as input for Algorithm~\ref{alg: Delayer strat} upon the next query of the Prover.

    \begin{algorithm}[H]
        \caption{Delayer strategy for a reduced AND-OR tree $\T$}\label{alg: Delayer strat}
        \begin{algorithmic}[1]
            \Require $\T$, $i$
            \State $\ell \gets $ leaf in $\T$ such that $\index(\ell) = i$.
            \If{$\ell$ is the only leaf in $\T$}
            \State $b \gets $ Prover's choice.
            \ElsIf{$\llabel(\parent(\ell)) = \lor$}
            \State $b \gets 0$.
            \Else
            \State $v \gets \parent(\ell)$.
            \If{$|\child(v)| > 2$}
            \State $b \gets 1$.
            \Else
            \If{$v = \root(\T)$}
            \State $b \gets 1$.
            \Else
            \State $\{m\} \gets \child(v) \setminus \{\ell\}$.
            \If{$\exists x \in \child(m) : \llabel(x) = \land$}
            \State $b \gets 1$.
            \Else
            \State $u \gets \parent(v)$.
            \If{$\forall x \in \child(u) \setminus \{v\}$, $\llabel(x) = \textrm{leaf}$}
            \State $b \gets 1$.
            \Else
            \State $b \gets $ Prover's choice.
            \EndIf
            \EndIf
            \EndIf
            \EndIf
            \EndIf
            \Ensure $\contract(\update(\T,i,b))$.
        \end{algorithmic}
    \end{algorithm}

    Next, we argue that the Delayer always earns exactly $(n+2)/3$ points with this strategy. To that end, we introduce the following progress measure.
    \begin{definition}\label{def:progress-measure-lb}
        Let $\T$ be an AND-OR tree (not necessarily in reduced form). We define the following cost function recursively on the nodes of $\T$, as
        \[c(v) = \begin{cases}
            0, & \text{if } \llabel(v) = \leaf, \\
            c(w), & \text{if } \child(v) = \{w\}, \\
            1, & \text{if } \llabel(v) = \land, |\child(v)| \geq 2, \\
            \displaystyle \sum_{w \in \child(v)} c(w), & \text{if } \llabel(v) = \lor.
        \end{cases}\]
        We also define the set of marked nodes $M$, as
        \begin{align*}
            M =\;& \{v \in \T : \llabel(v) = \lor \text{ and } |\child(v)| \geq 2 \\
            &\qquad \text{ and } (\pparent(v) = \mathsf{undefined} \text{ or } \llabel(\pparent(v)) = \land)\}.
        \end{align*}
        We now define a \textit{progress measure} of $\T$, as
        \[P(\T) = \begin{cases}
            \displaystyle 1 + \sum_{v \in M} \max\{c(v) - 1, 0\}, & \text{if}~\T~ \text{is non-empty}, \\
            0, & \text{otherwise}.
        \end{cases}\]
    \end{definition}

    We observe that $c(v)$ only recursively depends on $\cbra{c(w) : w \in \child(v)}$. Consequently, if we prove that the cost function is not altered on all children of $v$ then $c(v)$ is also not altered. In particular, it is completely determined by the subtree rooted at $v$.

    Now, we show that as long as the progress function is positive, the game does not end.

    \begin{claim}
        \label{claim:PositiveP-implies-NonEmptyT}
        If $P(\T)>0$, then the function computed by $\T$ is not a constant function.
    \end{claim}

    \begin{proof}
        If $P(\T)>0$, then $\T$ is not empty. Thus, there is at least one leaf (variable) in the tree. Setting values of all leaves to $0$ causes $\T$ to evaluate to $0$, and setting values of all leaves to $1$ causes $\T$ to evaluate to $1$.
    \end{proof}

    Next, we show that the $\contract$ procedure does not affect the progress measure.

    \begin{claim}
        \label{claim:Contract-Changes-No-P}
        The $\contract$ procedure in Algorithm~\ref{alg: contract} does not alter the value of $P(\T)$.
    \end{claim}

    \begin{proof}
        The first part of the $\contract$ subroutine, i.e., Lines 1 to 4, looks for an internal node $v$ that has exactly one child $w$ and contracts it into a new node $v'$. Let $\T$ be the AND-OR tree before the contraction, and let $\T'$ be the tree afterwards. Similarly, let $c$ and $c'$ be the cost functions defined on $\T$ and $\T'$, respectively, and let $M$ and $M'$ be their marked sets. Since the cost function evaluated at a node is completely determined by the subtree rooted at it, we observe that $c'(v') = c(w)$. Moreover, it is easy to verify that both when $\llabel(v) = \lor$ and when $\llabel(v) = \land$, we have that $c(v) = c(w) = c'(v')$. It follows that the cost function on the rest of the tree is not changed anywhere. Finally, $v \not\in M$ since $|\child(v)| = 1$. Furthermore, $w \in M$ if and only if $v' \in M$, since this only depends on the proper parent of $v$ and $v'$ and their children, which are the same. Thus, $P(\T) = P(\T')$.

        Next, we focus on the second part of the $\contract$ routine, which merges two nodes that have the same label. Since the first part of the algorithm has been completed, we can assume that all vertices have at least two children, which means that the proper parent of any internal vertex is its direct parent. Now, let $\T$, $\T'$, $c$, $c'$, $M$ and $M'$ be as before, and let $v$ and $w$ be the parent and child nodes in $\T$, and let $v'$ be the new node in $\T'$. We immediately observe that $w \not\in M$, since it has the same label as its parent (which is its proper parent since all nodes have at least 2 children after the first phase of Algorithm~\ref{alg: contract}), and $v \in M \Leftrightarrow v' \in M'$, since both have at least two children, and hence it only depends on their parents, which are the same.
        \begin{itemize}
            \item	Suppose that $v$ and $w$ are both labeled by $\lor$. Then,
            \begin{align*}
                c'(v') &= \sum_{x \in \child(v')} c'(x) \tag*{by definition of the cost function and $\llabel(v') = \lor$} \\
                &= \sum_{x \in \child(v) \setminus \{w\}} c'(x) + \sum_{x \in \child(w)} c'(x) \tag*{by Lines~6-7 of Algorithm~\ref{alg: contract}} \\
                &= \sum_{x \in \child(v) \setminus \{w\}} c(x) + \sum_{x \in \child(w)} c(x) \tag*{since $c(x)$ only depends on $c(y)$ for $y \in \child(x)$} \\
                &= c(w) + \sum_{x \in \child(v) \setminus \{w\}} c(x) \tag*{by definition of the cost function, and since $\llabel(w) = \lor$} \\
                &= \sum_{x \in \child(v)} c(x) \\
                &= c(v),
            \end{align*}
            where the last equality follows from the definition of the cost function in Definition~\ref{def:progress-measure-lb}, and since the label of $v$ was assumed to be $\lor$. Since the cost function of any given node depends only on its immediate children it follows that the cost function is not changed in any other part of the tree.
            \item Similarly, if $v$ and $w$ are both labeled by $\land$, then since both $v$ and $w$ have more than $1$ child, so does $v'$, and hence $c(v) = c'(v') = 1$, and thus in this case the cost function is not changed anywhere else either.
        \end{itemize}
        Thus, the $\contract$ procedure leaves $P(\T)$ invariant.
    \end{proof}

    \begin{claim}
        \label{claim:PointScored-implies-Pdecreases}
        In Algorithm~\ref{alg: Delayer strat}, if the score is increased by 1, then $P(\T)$ is decreased by 1. At any other point in the algorithm, the value of $P(\T)$ does not change. In other words, $\score + P(\T)$ is constant throughout the Prover-Delayer game.
    \end{claim}

    \begin{proof}
        We already know from Claim~\ref{claim:Contract-Changes-No-P} that the $\contract$ procedure does not impact the progress measure $P(\T)$. Thus, it remains to check that the $\update$ procedure decreases the progress measure by 1 if and only if the Delayer scores a point. To that end, let $\T$ be the tree at the start of Algorithm~\ref{alg: Delayer strat}, and let $\T'$ be the tree after the $\update$ procedure is performed, but before the $\contract$ procedure is called, in the final line of Algorithm~\ref{alg: Delayer strat}. Let $c$, $c'$, $M$ and $M'$ be the cost functions and marked sets of $\T$ and $\T'$, respectively.

        Below, we go through the same case analysis as presented in Algorithm~\ref{alg: Delayer strat}. Also see Figure~\ref{fig: case analysis} for a graphical description of each case.

        \begin{enumerate}
            \item Suppose that $\T$ contains only one leaf, which the Prover queries. We have $P(\T) = 1$, the Delayer scores a point in Line~3 of the Algorithm~\ref{alg: Delayer strat}, and the update routine makes the tree empty in Line~3 of Algorithm~\ref{alg: update}. Thus, we find that $\T'$ is empty, and hence $P(\T) - P(\T') = 1 - 0 = 1$.
            \item Suppose that the parent of the leaf being queried, denoted by $v$ as in Line~7 of Algorithm~\ref{alg: Delayer strat}, is labeled by $\lor$. Then, the Delayer sets the corresponding variable to $0$ in Line~5 of Algorithm~\ref{alg: Delayer strat}, which means that Line~5 of Algorithm~\ref{alg: update} removes $\ell$ from the tree. We find that $c(v) - c'(v) = c(\ell) = 0$, and hence $c|_{\T'} = c'$. Moreover, if $v$ had more than $2$ children in $\T$, then $M = M'$, and hence $P(\T) = P(\T')$. On the other hand, if $v$ has exactly $2$ children in $\T$, then $v$ will have only one child in $\T'$ and hence $v \not\in M'$. Let $w$ be the other child of $v$ in $\T$. Since $\T$ is in reduced form, $w$ is either a leaf or its label is $\land$, wich implies that $w \not\in M \cap M'$. Moreover, if $w$ is an internal node, it must have at least $2$ children, and hence any node in the subtree rooted at $w$ will have a proper parent that is also contained in that subtree. Therefore, $M = M'$, and since the cost function evaluated at any given root only depends on the tree rooted at the given node, we also find that $c'(v) = c'(w) = c(w) = c(v)$. Thus, $P(\T) = P(\T')$, and hence the progress measure is not affected.
            \item Next we analyze the case where the parent node $v$ has label $\land$. Since $\T$ is in reduced form, $v$ has at least $2$ children, and hence $c(v) = 1$. Suppose that $v$ has more than $2$ children. Then, Line~9 of Algorithm~\ref{alg: Delayer strat} sets the corresponding variable to $1$, which means that Line~5 of Algorithm~\ref{alg: update} removes $\ell$ from the tree. Since the remaining number of children of $v$ is still at least $2$, we have that $c'(v) = c(v) = 1$, and so also on the remainder of the tree, the cost function is not changed. Similarly, $M = M'$, and hence the $P(\T') = P(\T)$.
            \item We next consider the case where $v$ has exactly $2$ children, and we let $m$ be the other child of $v$, as in Line~14 of Algorithm~\ref{alg: Delayer strat}. If $v$ is the root of $\T$, then Line~12 in Algorithm~\ref{alg: Delayer strat} sets the variable associated to the queried leaf to $1$, which means that Line~5 in Algorithm~\ref{alg: update} removes the queried leaf from the tree. The cost function on all nodes in $\T$ apart from $v$ will not be altered, and $v \not\in M \cup M'$ and $m \in M \cap M'$. Thus, $P(\T) = P(\T')$, proving that the progress measure is left unaltered.
            \item It remains to check the case where $v$ is not the root of $\T$, in which case we can let $u = \parent(v)$. Since $\T$ is in reduced form, we have that the label of $u$ must be different from the label of $v$, so we find that $\llabel(u) = \lor$. Suppose that $m$ is an internal vertex and that at least one child of $m$ is labeled by $\land$, then we have $c(m) \geq 1$, and similarly $c(u) \geq 1$. Line~16 of Algorithm~\ref{alg: Delayer strat} now sets the queried leaf to $1$, which means that Line~5 in Algorithm~\ref{alg: update} removes the queried leaf from the tree. Since $v$ is now left with just one child, $c'(v) = c'(m) = c(m)$. This implies that $c'(u) - c(u) = c'(v) - c(v) = c(m) - 1$, and we also find that $u,m \in M$ and $u \in M'$, but $m \not\in M'$. Thus,
            \begin{align*}
                \P(\T) - \P(\T') &= \max\{c(u)-1,0\} + \max\{c(m)-1,0\} - \max\{c'(u)-1,0\} \\
                &= c(u)-1 + c(m)-1 - (c'(u)-1) \\
                &= c(u) - c'(u) + c(m) - 1 = 0,
            \end{align*}
            from which we observe that the progress measure is indeed not altered.
            \item It remains to check the case where $m$ does not have a child labeled by $\land$, which implies that it is a leaf itself or all its children must be leaves, and thus $c(m) = 0$. Suppose that all children of $u$, except for $v$, are leaves, which means that $c(u) = c(v) = 1$. Then, Line~20 of Algorithm~\ref{alg: Delayer strat} sets the variable corresponding to the queried leaf to $1$, which means that Line~5 of Algorithm~\ref{alg: update} removes it from the tree. Then, we find that $c'(u) = c'(v) = c'(m) = c(m) = 0$, and since $\T$ is in reduced form, the parent of $u$, if it exists, must be labeled by $\land$ and have multiple children, which means that the cost function is not altered anywhere outside of the subtree rooted at $u$. Finally, since
            \begin{align*}
                \max\{c(u)-1,0\} = \max\{c(m)-1,0\} = \max\{c'(u)-1,0\} = \max\{c'(m)-1,0\} = 0,
            \end{align*}
            we conclude that none of the vertices in the subtree rooted at $u$ makes any contribution to the progress measure, and hence it remains unaffected.
            \item Thus, it remains to investigate the case where $u$ has multiple children that are labeled by $\land$, which means that $c(u) \geq 2$. In that case, the Delayer scores a point in Line~22 of Algorithm~\ref{alg: Delayer strat}. If the Prover chooses the value $1$, then Line~5 of Algorithm~\ref{alg: update} removes the leaf from the tree. Since $c(u) \geq 2$, and $c(v) = 1$, we find that $c'(u) \geq 1$. Since $\T$ is in reduced form, we find that the label of the parent of $u$, if it exists, must be $\land$, and it must have multiple children, from which we infer that the cost function remains unaffected on all of $\T$ except for the subtree rooted at $u$. Furthermore, we have $u \in M \cap M'$, and $\max\{c(m)-1,0\} = \max\{c'(m)-1,0\} = 0$, which means that
            \begin{align*}
                P(\T) - P(\T') &= \max\{c(u)-1,0\} - \max\{c'(u)-1,0\} \\
                &= c(u) - c'(u) = c(v) - c'(v) = 1 - c'(m) = 1 - c(m) = 1,
            \end{align*}
            which indeed proves that the progress measure is decreased by $1$.
            \item On the other hand, if the Prover chooses $0$, then Line~7 in Algorithm~\ref{alg: update} removes the leaf from the tree, as well as the node $v$ and the subtree rooted at it. Again, the cost function is not altered on all of $\T$, except for the subtree rooted at $u$. If $u$ has more than $2$ children in $\T$, then, we find that $u \in M$ and $u \in M'$, which implies
            \begin{align*}
                P(\T) - P(\T') &= \max\{c(u)-1,0\} - \max\{c'(u)-1,0\} = c(u) - c'(u) = c(v) = 1,
            \end{align*}
            and if $u$ has exactly $2$ children in $\T$, then $c(u) = 2$, $u \in M$ and $u \not\in M'$, from which we find that
            \begin{align*}
                P(\T) - P(\T') &= \max\{c(u)-1,0\} = 1.
            \end{align*}
        \end{enumerate}
        Thus, indeed, the progress measure is decreased by $1$ iff the Prover scores a point (which happened in Items 1, 7 and 8 above). This completes the proof.
    \end{proof}

    \begin{figure}[p]
        \centering
        \begin{tabular}{|c|c|c|c|}
            \hline
            & $\T$ && $\T'$ \\
            \hline
            Case 1 & \raisebox{-.5em}{\begin{tikzpicture}\small
                    \fill (0,0) circle[radius=.15em];
                    \node[below] at (0,0) {$x_i$};
            \end{tikzpicture}} & $x_i = b$ & $b$ \\
            \hline
            Case 2 & \raisebox{-2.5em}{\begin{tikzpicture}[scale=.6,gate/.style={draw,circle,minimum size=1.2em,inner sep=0}]\small
                    \node (leaf) at (-.5,0) {$x_i$};
                    \node[gate] (OR) at (1,1) {$\lor$};
                    \node[above right=.3em] at (1,1) {$v$};
                    \node (parent) at (1,2) {$\triangle$};
                    \node (firstneighbor) at (.5,0) {$\triangle$};
                    \node at (1.5,0) {$\cdots$};
                    \node (lastneighbor) at (2.5,0) {$\triangle$};
                    \draw[->] (OR) to (leaf);
                    \draw[->] (parent) to (OR);
                    \draw[->] (OR) to (firstneighbor);
                    \draw[->] (OR) to (lastneighbor);
            \end{tikzpicture}} & $x_i = 0$ & \raisebox{-2.5em}{\begin{tikzpicture}[scale=.6,gate/.style={draw,circle,minimum size=1.2em,inner sep=0}]\small
                    \node[gate] (OR) at (1,1) {$\lor$};
                    \node[above right=.3em] at (1,1) {$v$};
                    \node (parent) at (1,2) {$\triangle$};
                    \node (firstneighbor) at (.5,0) {$\triangle$};
                    \node at (1.5,0) {$\cdots$};
                    \node (lastneighbor) at (2.5,0) {$\triangle$};
                    \draw[->] (parent) to (OR);
                    \draw[->] (OR) to (firstneighbor);
                    \draw[->] (OR) to (lastneighbor);
            \end{tikzpicture}} \\
            \hline
            Case 3 & \raisebox{-2.7em}{\begin{tikzpicture}[scale=.6,gate/.style={draw,circle,minimum size=1.2em,inner sep=0}]\small
                    \node (leaf) at (-.5,0) {$x_i$};
                    \node[gate] (AND) at (1,1) {$\land$};
                    \node[above right=.3em] at (1,1) {$v$};
                    \node (parent) at (1,2) {$\triangle$};
                    \node (firstneighbor) at (.5,0) {$\triangle$};
                    \node at (1.5,0) {$\cdots$};
                    \node (lastneighbor) at (2.5,0) {$\triangle$};
                    \draw[decorate, decoration={calligraphic brace,mirror,raise=6pt}] (.3,0) to node[below=6pt] {$\geq 2$ children} (2.7,0);
                    \draw[->] (AND) to (leaf);
                    \draw[->] (parent) to (AND);
                    \draw[->] (AND) to (firstneighbor);
                    \draw[->] (AND) to (lastneighbor);
            \end{tikzpicture}} & $x_i = 1$ & \raisebox{-2.7em}{\begin{tikzpicture}[scale=.6,gate/.style={draw,circle,minimum size=1.2em,inner sep=0}]\small
                    \node[gate] (AND) at (1,1) {$\land$};
                    \node[above right=.3em] at (1,1) {$v$};
                    \node (parent) at (1,2) {$\triangle$};
                    \node (firstneighbor) at (.5,0) {$\triangle$};
                    \node at (1.5,0) {$\cdots$};
                    \node (lastneighbor) at (2.5,0) {$\triangle$};
                    \draw[decorate, decoration={calligraphic brace,mirror,raise=6pt}] (.3,0) to node[below=6pt] {$\geq 2$ children} (2.7,0);
                    \draw[->] (parent) to (AND);
                    \draw[->] (AND) to (firstneighbor);
                    \draw[->] (AND) to (lastneighbor);
            \end{tikzpicture}} \\
            \hline
            Case 4 & \raisebox{-2.5em}{\begin{tikzpicture}[scale=.6,gate/.style={draw,circle,minimum size=1.2em,inner sep=0}]\small
                    \node (leaf) at (0,0) {$x_i$};
                    \node[gate] (AND) at (1,1) {$\land$};
                    \node[above right=.3em] at (1,1) {$v$};
                    \node[gate] (m) at (2,0) {$\lor$};
                    \node[above right=.3em] at (m) {$m$};
                    \node (firstneighbor) at (1,-1) {$\triangle$};
                    \node at (2,-1) {$\cdots$};
                    \node (lastneighbor) at (3,-1) {$\triangle$};
                    \draw[->] (AND) to (leaf);
                    \draw[->] (AND) to (m);
                    \draw[->] (m) to (firstneighbor);
                    \draw[->] (m) to (lastneighbor);
            \end{tikzpicture}} & $x_i = 1$ & \raisebox{-2.5em}{\begin{tikzpicture}[scale=.6,gate/.style={draw,circle,minimum size=1.2em,inner sep=0}]\small
                    \node[gate] (AND) at (1,1) {$\land$};
                    \node[above right=.3em] at (1,1) {$v$};
                    \node[gate] (m) at (2,0) {$\lor$};
                    \node[above right=.3em] at (m) {$m$};
                    \node (firstneighbor) at (1,-1) {$\triangle$};
                    \node at (2,-1) {$\cdots$};
                    \node (lastneighbor) at (3,-1) {$\triangle$};
                    \draw[->] (AND) to (m);
                    \draw[->] (m) to (firstneighbor);
                    \draw[->] (m) to (lastneighbor);
            \end{tikzpicture}} \\
            \hline
            Case 5 & \raisebox{-4.5em}{\begin{tikzpicture}[scale=.6,gate/.style={draw,circle,minimum size=1.2em,inner sep=0}]\small
                    \node (leaf) at (0,0) {$x_i$};
                    \node[gate] (AND) at (1,1) {$\land$};
                    \node[above left=.3em] at (1,1) {$v$};
                    \node[gate] (m) at (2,0) {$\lor$};
                    \node[above right=.3em] at (m) {$m$};
                    \node[gate] (u) at (2.5,2) {$\lor$};
                    \node[above right=.3em] at (u) {$u$};
                    \node (top) at (2.5,3) {$\triangle$};
                    \node (firstupneighbor) at (2,1) {$\triangle$};
                    \node at (3,1) {$\cdots$};
                    \node (lastupneighbor) at (4,1) {$\triangle$};
                    \node[gate] (childand) at (.5,-1) {$\land$};
                    \node (firstneighbor) at (1.5,-1) {$\triangle$};
                    \node at (2.5,-1) {$\cdots$};
                    \node (lastneighbor) at (3.5,-1) {$\triangle$};
                    \node (bottomleftchild) at (-.5,-2) {$\triangle$};
                    \node at (.5,-2) {$\cdots$};
                    \node (bottomrightchild) at (1.5,-2) {$\triangle$};
                    \draw[->] (top) to (u);
                    \draw[->] (AND) to (leaf);
                    \draw[->] (AND) to (m);
                    \draw[->] (u) to (AND);
                    \draw[->] (u) to (firstupneighbor);
                    \draw[->] (u) to (lastupneighbor);
                    \draw[->] (m) to (firstneighbor);
                    \draw[->] (m) to (lastneighbor);
                    \draw[->] (m) to (childand);
                    \draw[->] (childand) to (bottomleftchild);
                    \draw[->] (childand) to (bottomrightchild);
            \end{tikzpicture}} & $x_i = 1$ & \raisebox{-4.5em}{\begin{tikzpicture}[scale=.6,gate/.style={draw,circle,minimum size=1.2em,inner sep=0}]\small
                    \node[gate] (AND) at (1,1) {$\land$};
                    \node[above left=.3em] at (1,1) {$v$};
                    \node[gate] (m) at (2,0) {$\lor$};
                    \node[above right=.3em] at (m) {$m$};
                    \node[gate] (u) at (2.5,2) {$\lor$};
                    \node[above right=.3em] at (u) {$u$};
                    \node (top) at (2.5,3) {$\triangle$};
                    \node (firstupneighbor) at (2,1) {$\triangle$};
                    \node at (3,1) {$\cdots$};
                    \node (lastupneighbor) at (4,1) {$\triangle$};
                    \node[gate] (childand) at (.5,-1) {$\land$};
                    \node (firstneighbor) at (1.5,-1) {$\triangle$};
                    \node at (2.5,-1) {$\cdots$};
                    \node (lastneighbor) at (3.5,-1) {$\triangle$};
                    \node (bottomleftchild) at (-.5,-2) {$\triangle$};
                    \node at (.5,-2) {$\cdots$};
                    \node (bottomrightchild) at (1.5,-2) {$\triangle$};
                    \draw[->] (top) to (u);
                    \draw[->] (AND) to (m);
                    \draw[->] (u) to (AND);
                    \draw[->] (u) to (firstupneighbor);
                    \draw[->] (u) to (lastupneighbor);
                    \draw[->] (m) to (firstneighbor);
                    \draw[->] (m) to (lastneighbor);
                    \draw[->] (m) to (childand);
                    \draw[->] (childand) to (bottomleftchild);
                    \draw[->] (childand) to (bottomrightchild);
            \end{tikzpicture}} \\
            \hline
            Case 6 & \raisebox{-3.5em}{\begin{tikzpicture}[scale=.6,gate/.style={draw,circle,minimum size=1.2em,inner sep=0}]\small
                    \node (leaf) at (0,0) {$x_i$};
                    \node[gate] (AND) at (1,1) {$\land$};
                    \node[above left=.3em] at (1,1) {$v$};
                    \node[gate] (m) at (2,0) {$\lor$};
                    \node[above right=.3em] at (m) {$m$};
                    \node[gate] (u) at (2.5,2) {$\lor$};
                    \node[above right=.3em] at (u) {$u$};
                    \node (top) at (2.5,3) {$\triangle$};
                    \node (firstupneighbor) at (2,1) {$\bullet$};
                    \node at (3,1) {$\cdots$};
                    \node (lastupneighbor) at (4,1) {$\bullet$};
                    \node (firstneighbor) at (1,-1) {$\bullet$};
                    \node at (2,-1) {$\cdots$};
                    \node (lastneighbor) at (3,-1) {$\bullet$};
                    \draw[->] (top) to (u);
                    \draw[->] (AND) to (leaf);
                    \draw[->] (AND) to (m);
                    \draw[->] (u) to (AND);
                    \draw[->] (u) to (firstupneighbor);
                    \draw[->] (u) to (lastupneighbor);
                    \draw[->] (m) to (firstneighbor);
                    \draw[->] (m) to (lastneighbor);
            \end{tikzpicture}} & $x_i = 1$ & \raisebox{-3.5em}{\begin{tikzpicture}[scale=.6,gate/.style={draw,circle,minimum size=1.2em,inner sep=0}]\small
                    \node[gate] (AND) at (1,1) {$\land$};
                    \node[above left=.3em] at (1,1) {$v$};
                    \node[gate] (m) at (2,0) {$\lor$};
                    \node[above right=.3em] at (m) {$m$};
                    \node[gate] (u) at (2.5,2) {$\lor$};
                    \node[above right=.3em] at (u) {$u$};
                    \node (top) at (2.5,3) {$\triangle$};
                    \node (firstupneighbor) at (2,1) {$\bullet$};
                    \node at (3,1) {$\cdots$};
                    \node (lastupneighbor) at (4,1) {$\bullet$};
                    \node (firstneighbor) at (1,-1) {$\bullet$};
                    \node at (2,-1) {$\cdots$};
                    \node (lastneighbor) at (3,-1) {$\bullet$};
                    \draw[->] (top) to (u);
                    \draw[->] (AND) to (m);
                    \draw[->] (u) to (AND);
                    \draw[->] (u) to (firstupneighbor);
                    \draw[->] (u) to (lastupneighbor);
                    \draw[->] (m) to (firstneighbor);
                    \draw[->] (m) to (lastneighbor);
            \end{tikzpicture}} \\
            \hline
            Case 7 & \raisebox{-3.5em}{\begin{tikzpicture}[scale=.6,gate/.style={draw,circle,minimum size=1.2em,inner sep=0}]\small
                    \node (leaf) at (0,0) {$x_i$};
                    \node[gate] (AND) at (1,1) {$\land$};
                    \node[above left=.3em] at (1,1) {$v$};
                    \node[gate] (m) at (2,0) {$\lor$};
                    \node[above right=.3em] at (m) {$m$};
                    \node[gate] (u) at (2.5,2) {$\lor$};
                    \node[above right=.3em] at (u) {$u$};
                    \node (top) at (2.5,3) {$\triangle$};
                    \node (firstupneighbor) at (2,1) {$\triangle$};
                    \node at (3,1) {$\cdots$};
                    \node (lastupneighbor) at (4,1) {$\triangle$};
                    \node[gate] (andchild) at (5,1) {$\land$};
                    \node (left) at (4,0) {$\triangle$};
                    \node at (5,0) {$\cdots$};
                    \node (right) at (6,0) {$\triangle$};
                    \node (firstneighbor) at (1,-1) {$\bullet$};
                    \node at (2,-1) {$\cdots$};
                    \node (lastneighbor) at (3,-1) {$\bullet$};
                    \draw[->] (top) to (u);
                    \draw[->] (AND) to (leaf);
                    \draw[->] (AND) to (m);
                    \draw[->] (u) to (AND);
                    \draw[->] (u) to (firstupneighbor);
                    \draw[->] (u) to (lastupneighbor);
                    \draw[->] (u) to (andchild);
                    \draw[->] (m) to (firstneighbor);
                    \draw[->] (m) to (lastneighbor);
                    \draw[->] (andchild) to (left);
                    \draw[->] (andchild) to (right);
            \end{tikzpicture}} & $x_i = 1$ & \raisebox{-3.5em}{\begin{tikzpicture}[scale=.6,gate/.style={draw,circle,minimum size=1.2em,inner sep=0}]\small
                    \node[gate] (AND) at (1,1) {$\land$};
                    \node[above left=.3em] at (1,1) {$v$};
                    \node[gate] (m) at (2,0) {$\lor$};
                    \node[above right=.3em] at (m) {$m$};
                    \node[gate] (u) at (2.5,2) {$\lor$};
                    \node[above right=.3em] at (u) {$u$};
                    \node (top) at (2.5,3) {$\triangle$};
                    \node (firstupneighbor) at (2,1) {$\triangle$};
                    \node at (3,1) {$\cdots$};
                    \node (lastupneighbor) at (4,1) {$\triangle$};
                    \node[gate] (andchild) at (5,1) {$\land$};
                    \node (left) at (4,0) {$\triangle$};
                    \node at (5,0) {$\cdots$};
                    \node (right) at (6,0) {$\triangle$};
                    \node (firstneighbor) at (1,-1) {$\bullet$};
                    \node at (2,-1) {$\cdots$};
                    \node (lastneighbor) at (3,-1) {$\bullet$};
                    \draw[->] (top) to (u);
                    \draw[->] (AND) to (m);
                    \draw[->] (u) to (AND);
                    \draw[->] (u) to (firstupneighbor);
                    \draw[->] (u) to (lastupneighbor);
                    \draw[->] (u) to (andchild);
                    \draw[->] (m) to (firstneighbor);
                    \draw[->] (m) to (lastneighbor);
                    \draw[->] (andchild) to (left);
                    \draw[->] (andchild) to (right);
            \end{tikzpicture}} \\
            \hline
            Case 8 & \raisebox{-3.5em}{\begin{tikzpicture}[scale=.6,gate/.style={draw,circle,minimum size=1.2em,inner sep=0}]\small
                    \node (leaf) at (0,0) {$x_i$};
                    \node[gate] (AND) at (1,1) {$\land$};
                    \node[above left=.3em] at (1,1) {$v$};
                    \node[gate] (m) at (2,0) {$\lor$};
                    \node[above right=.3em] at (m) {$m$};
                    \node[gate] (u) at (2.5,2) {$\lor$};
                    \node[above right=.3em] at (u) {$u$};
                    \node (top) at (2.5,3) {$\triangle$};
                    \node (firstupneighbor) at (2,1) {$\triangle$};
                    \node at (3,1) {$\cdots$};
                    \node (lastupneighbor) at (4,1) {$\triangle$};
                    \node[gate] (andchild) at (5,1) {$\land$};
                    \node (left) at (4,0) {$\triangle$};
                    \node at (5,0) {$\cdots$};
                    \node (right) at (6,0) {$\triangle$};
                    \node (firstneighbor) at (1,-1) {$\bullet$};
                    \node at (2,-1) {$\cdots$};
                    \node (lastneighbor) at (3,-1) {$\bullet$};
                    \draw[->] (top) to (u);
                    \draw[->] (AND) to (leaf);
                    \draw[->] (AND) to (m);
                    \draw[->] (u) to (AND);
                    \draw[->] (u) to (firstupneighbor);
                    \draw[->] (u) to (lastupneighbor);
                    \draw[->] (u) to (andchild);
                    \draw[->] (m) to (firstneighbor);
                    \draw[->] (m) to (lastneighbor);
                    \draw[->] (andchild) to (left);
                    \draw[->] (andchild) to (right);
            \end{tikzpicture}} & $x_i = 0$ & \raisebox{-2.5em}{\begin{tikzpicture}[scale=.6,gate/.style={draw,circle,minimum size=1.2em,inner sep=0}]\small
                    \node[gate] (u) at (2.5,2) {$\lor$};
                    \node[above right=.3em] at (u) {$u$};
                    \node (top) at (2.5,3) {$\triangle$};
                    \node (firstupneighbor) at (2,1) {$\triangle$};
                    \node at (3,1) {$\cdots$};
                    \node (lastupneighbor) at (4,1) {$\triangle$};
                    \node[gate] (andchild) at (5,1) {$\land$};
                    \node (left) at (4,0) {$\triangle$};
                    \node at (5,0) {$\cdots$};
                    \node (right) at (6,0) {$\triangle$};
                    \draw[->] (top) to (u);
                    \draw[->] (u) to (firstupneighbor);
                    \draw[->] (u) to (lastupneighbor);
                    \draw[->] (u) to (andchild);
                    \draw[->] (andchild) to (left);
                    \draw[->] (andchild) to (right);
            \end{tikzpicture}} \\\hline
        \end{tabular}
        \caption{Graphical depiction of the cases considered in the proof of Claim~\ref{claim:PointScored-implies-Pdecreases}. Black dots indicate leaves, and triangles depict arbitrary AND-OR trees. In cases 6 to 8, the subtree rooted at $m$ can also be replaced by a single leaf.}
        \label{fig: case analysis}
    \end{figure}

    Now that we know that the Delayer's strategy ensures that $\score + P(\T)$ remains constant throughout the Prover-Delayer game, we observe that the initial value of the progress measure tells us how many points will be scored throughout the game.

    \begin{claim}
        \label{claim:Initial-P-value}
        At the start of the game, the progress measure is $(n+2)/3$.
    \end{claim}
    \begin{proof}
        At the beginning of Algorithm~\ref{alg: Delayer strat} the tree $\T$ is a complete binary tree with every node $v \in \T$ with $\llabel(v)=\lor$ has exactly two children, say $w_1$ and $w_2$, with $\llabel(w_1)=\llabel(w_2)=\land$. Since $w_1$ and $w_2$ both have two children as well, we have $c_{\land}(w_1) = c_{\land}(w_2) = 1$, which implies that $\max\{c_{\land}(v)-1,0\} = 2$. Therefore, the progress measure is equal to $1$ plus the number of $\lor$-labelled nodes, which can be easily seen to equal (with $4^k = n$),
        \[1 + \sum_{j=1}^k \frac{n}{4^j} = 1 + \frac{n}{4} \sum_{j=0}^{k-1} \frac{1}{4^j} = 1 + \frac{n}{4} \frac{1 - \frac{1}{4^k}}{1 - \frac14} = 1 + \frac{n}{4} \cdot \frac{1 - \frac1n}{\frac34} = 1 + \frac{n-1}{3} = \frac{n+2}{3},\]
        completing the proof.
    \end{proof}

    \begin{claim}
        \label{claim:rank-lb}
        The rank of $F$ is at least $(n+2)/3$.
    \end{claim}

    \begin{proof}
        The progress measure starts at $(n+2)/3$ (Claim~\ref{claim:Initial-P-value}), decreases by exactly $1$ whenever a point is scored (Claim~\ref{claim:PointScored-implies-Pdecreases}), and the game doesn't end until it is $0$ (Claim~\ref{claim:PositiveP-implies-NonEmptyT}). This proves that the Delayer can at least score $(n+2)/3$ points. Thus, by using Claim~\ref{claim: rank equals Prover Delayer game value}, we see that the rank is indeed bounded from below by $(n+2)/3$.
    \end{proof}

    \subsubsection{Prover's Strategy (Upper Bound)} The Prover also starts by building the complete AND-OR tree $\T$ on $n$ bits. Then, in every round of the Prover-Delayer game, the Prover asks the variable that is associated to an arbitrary leaf of the tree $\T$. Next, if the Delayer answers $0$ or $1$, the Prover updates the tree with the $\update$ procedure, Algorithm~\ref{alg: update}. If instead, the Delayer asks the Prover to choose, then the Prover always chooses $0$, and updates the tree accordingly. Finally, the Prover puts the tree back into reduced form using the $\contract$ procedure. The algorithm is formally stated in Algorithm~\ref{alg: Prover strat}.

    \begin{algorithm}[H]
        \caption{Prover's strategy}\label{alg: Prover strat}
        \begin{algorithmic}[1]
            \Require $\T$
            \State Query an arbitrary leaf of $\T$, associated with variable $i$.
            \State $b \gets $ answer from Delayer.
            \If{$b = $ Prover's choice}
            \State Choose $0$.
            \State $b \gets 0$.
            \EndIf
            \Ensure $\contract(\update(\T,i,b))$.
        \end{algorithmic}
    \end{algorithm}

    We now argue that if the Prover uses the strategy outlined in Algorithm~\ref{alg: Prover strat}, the Delayer cannot score more than $(n+2)/3$ points. To that end, we introduce the following progress measure:

    \begin{definition}\label{def:progress-measure-ub}
        Let $\T$ be an AND-OR tree (not necessarily in reduced form). We define the following cost function recursively on the nodes of $\T$, as
        \[d(v) = \begin{cases}
            \displaystyle \sum_{w \in \child(v)} d(w), & \text{if } \llabel(v) = \lor, \\
            d(w), & \text{if } \child(v) = \{w\}, \\
            1, & \text{otherwise}.
        \end{cases}\]
        We also define the set of marked nodes $M$ as before in Definition~\ref{def:progress-measure-lb}, as
        \begin{align*}
            M =\;& \{v \in \T : \llabel(v) = \lor \text{ and } |\child(v)| \geq 2 \\
            &\qquad \text{ and } (\pparent(v) = \mathsf{undefined} \text{ or } \llabel(\pparent(v)) = \land)\}.
        \end{align*}
        We now define a new \textit{progress measure} of $\T$, as
        \[S(\T) = \begin{cases}
            \displaystyle 1 + \sum_{v \in M} \max\{d(v) - 1, 0\}, & \text{if}~\T~ \text{is non-empty}, \\
            0, & \text{otherwise}.
        \end{cases}\]
    \end{definition}

    Note that the cost function introduced in Definition~\ref{def:progress-measure-ub} only differs slightly from the cost function introduced in Definition~\ref{def:progress-measure-lb}, in that leaves have cost 1 in Definition~\ref{def:progress-measure-ub}, whereas they have cost 0 in Definition~\ref{def:progress-measure-ub}. Just as before, we observe that $d$ evaluated at $v$ only depends on the immediate children of $v$, and whether $v$ belongs to $M$ only depends on its parent.

    \begin{claim}\label{claim:progmeasure0endgame}
        If the progress measure is $0$, the Prover-Delayer game ends.
    \end{claim}

    \begin{proof}
        If the progress measure is $0$, then the tree is empty, and hence the function value must be completely determined. This implies that the game ends, completing the proof.
    \end{proof}

    \begin{claim}\label{claim: contract-ub}
        The $\contract$ procedure does not alter the value of $S(\T)$.
    \end{claim}

    \begin{proof}
        The first part of the $\contract$ procedure looks for a node $v$ that has only one child $w$, and it contracts these into a new node $v'$. Let $\T$ be the tree before this contraction, and $\T'$ the resulting tree afterwards, and define the cost functions and marked sets $d$, $d'$, $M$ and $M'$ accordingly. We have $d(v) = d(w)$, both when $\llabel(v) = \lor$ and $\llabel(v) = \land$, and we also have $d(v') = d(w)$, since $\llabel(v') = \llabel(w)$ and the cost function only depends on their immediate children. Next observe that $v \notin M$ since $|\child(v)| = 1$, and $\pparent(w) = \pparent(v')$. Thus, $w \in M \iff v' \in M'$.
        Since the cost of other nodes and the marked set remain unaffected in the remainder of the tree, we find that $S(\T') = S(\T)$.

        In the second part of the $\contract$ procedure, we find a node $v$ and one of its children $w$ having the same label, and we contract them into a new node $v'$. Again let $\T$ be the tree before this contraction, and $\T'$ the resulting tree afterwards. Define the cost functions $d$, $d'$, $M$ and $M'$, accordingly.
        \begin{itemize}
            \item
            If $\llabel(v) = \llabel(w) = \lor$, then
            \begin{align*}
                d'(v') &= \sum_{x \in \child(v')} d'(x) \tag*{by definition of the cost function and $\llabel(v') = \lor$} \\
                &= \sum_{x \in \child(v) \setminus \{w\}} d'(x) + \sum_{x \in \child(w)} d'(x) \tag*{by Lines~6-7 in Algorithm~\ref{alg: contract}} \\
                &= \sum_{x \in \child(v) \setminus \{w\}} d(x) + \sum_{x \in \child(w)} d(x) \tag*{since $d(x)$ depends only on $d(y)$ for $y \in \child(x)$} \\
                &= \sum_{x \in \child(v) \setminus \{w\}} d(x) + d(w) \tag*{by definition of the cost function and $\llabel(w) = \lor$} \\
                &= d(v),
            \end{align*}
            where the final equality follows from the definition of the cost function and the fact that $\llabel(v) = \lor$. \item
            Similarly, if $\llabel(v) = \llabel(w) = \land$, then because $v$ and $w$ both have more than one child each, then so does $v'$, and hence $d'(v') = d(v) = 1$.
        \end{itemize}
        Thus, both in the cases where the labels are $\land$ and $\lor$, the cost function on the rest of the tree is not altered. Moreover, $w \notin M$ because its parent is $v$ (which is its proper parent since all nodes have at least 2 children after the first phase of Algorithm~\ref{alg: contract}), which has the same label as $w$. Moreover $v \in M \Leftrightarrow v' \in M'$, because their proper parents are the same.
        Thus, we find that $S(\T') = S(\T)$, completing the proof.
    \end{proof}

    \begin{claim}\label{claim:point-scored-S-decreases}
        Every time the Delayer asks the Prover to choose, the progress measure decreases by at least $1$.
    \end{claim}

    \begin{proof}
        We already know from Claim~\ref{claim: contract-ub} that the $\contract$ procedure does not impact the progress measure $S(\T)$. Thus, it remains to check that the $\update$ procedure decreases the progress measure by at least $1$, whenever the Delayer allows the Prover to choose. To that end, $\T$ be the tree at the moment the Delayer asks the Prover to choose. Since the Prover puts $\T$ back into reduced form after every round, we can assume it is in reduced form. According to the strategy outlined in Algorithm~\ref{alg: Prover strat}, the Prover chooses $0$. Let $\T'$ be the result of the $\update$ procedure, called in the last line of Algorithm~\ref{alg: Prover strat}, and let $d$, $d'$, $M$ and $M'$ be the associated cost functions and marked sets.

        Let the queried leaf be denoted by $\ell$. We consider three cases.

        \begin{enumerate}
            \item Suppose $\ell$ is the final remaining leaf in the tree. Then, Line~3 in the $\update$ procedure, i.e., Algorithm~\ref{alg: update}, will render the tree completely empty, and hence we find that $S(\T) - S(\T') = 1 - 0 = 1$. Thus, the progress measure indeed decreases by $1$.

            \item Let $v = \parent(\ell)$. Suppose $\llabel(v) = \lor$. Then, since the Prover always chooses $0$ according to Line~5 in Algorithm~\ref{alg: Prover strat}, we find that Line~5 of the $\update$ procedure, i.e., Algorithm~\ref{alg: update}, removes $\ell$ from the tree. Thus, we find that $d(v) - d'(v) = d(\ell) = 1$, by Definition~\ref{def:progress-measure-ub}. Moreover, $v \in M$ by Definition~\ref{def:progress-measure-ub} because $\llabel(v) = \lor$, $|\child(v)| \geq 2$, and $v$ is either the root node, or $\llabel(\pparent(v)) = \llabel(\parent(v)) = \land$ and $|\child(\parent(v))| \geq 2$ as $\T$ is in reduced form. Moreover, $d(v) \geq 2$, because
            $v$ has at least $2$ children that are either leaves or $\land$-nodes, and similarly $d'(v) \geq 1$. Thus, $\max\{d(v)-1,0\} = d(v)-1$ and $\max\{d'(v)-1,0\} = d'(v)-1$. Finally, the cost of other nodes and the set of marked nodes remain unaffected in the other parts of the tree, we find that $S(\T) - S(\T') = d(v) - d'(v) = 1$. Hence, also in this case the progress measure is decreased by $1$.

            \item On the other hand, suppose $\llabel(v) = \land$. Since the Prover chooses $0$, according to Line~5 of Algorithm~\ref{alg: Prover strat}, Line~7 of the $\update$ procedure, i.e., Algorithm~\ref{alg: update}, removes this node and the subtree rooted at $v$. If this node was the root, then the resulting tree is empty, meaning that $S(\T) - S(\T') \geq 1 - 0 = 1$, and hence the progress measure is indeed decreased by at least $1$. Otherwise, if $v$ was not the root, then let $u = \parent(v)$. We find that $\llabel(u) = \lor$, because $\T$ was in reduced form. Furthermore, since $|\child(u)| \geq 2$ in $\T$, we find that $d(u) \geq 2$, and $d'(u) \geq 1$. Furthermore, $v \not\in M$ because $\llabel(v) = \land$, and $u \in M$, because $\llabel(u) = \lor$, and the other conditions in the definition of $M$ are satisfied too. Moreover, the cost function above the node $u$ is not affected, and $M' \setminus\{u\} \subseteq M \setminus\{u\}$, where strict containment only happens whenever $M$ contains nodes in the subtree rooted at $v$ in $\T$.

            There are now two possible cases:
            \begin{itemize}
                \item If $|\child(u)| > 2$ in $\T$, then $|\child(u)| \geq 2$ in $\T'$, and hence we find that $u \in M'$ as well. Therefore, the difference between the old and new progress measures can be expressed as $S(\T) - S(\T') \geq d(u) - d'(u) = d(v) = 1$. Thus, indeed, the progress measure decreases by at least $1$.
                \item This leaves the case where $|\child(u)| = 2$ in $\T$, in which case we obtain that $d(u) = 2$ and $|\child(u)| = 1$ in $\T'$, and so $u \not\in M'$. This implies that $S(\T) - S(\T') \geq \max\{d(u) - 1, 0\} = 2 - 1 = 1$, and hence the progress measure decreases by at least $1$ in this case as well.
            \end{itemize}
        \end{enumerate}
        This completes the proof.
    \end{proof}

    \begin{claim}\label{claim:init-S-value}
        Initially, the progress measure $S(\T)$ is $(n+2)/3$.
    \end{claim}

    \begin{proof}
        Initially, every $\lor$-gate in the AND-OR tree is in $M$, because it is either the root node, or it has a parent node with $2$ children labeled by $\land$. Moreover, the cost of every node labeled by $\lor$ is $2$, because it has two children, both labeled by $\land$. Thus, we obtain that $S(\T)$ is $1$ plus the number of nodes labeled by $\lor$, which is $(n+2)/3$ (see the proof of Claim~\ref{claim:Initial-P-value}).
    \end{proof}

    \begin{claim}
        \label{claim:rank-ub}
        The rank of $F$ is at most $(n+2)/3$.
    \end{claim}

    \begin{proof}
        Since the progress measure starts at $(n+2)/3$ (Claim~\ref{claim:init-S-value}), with every point scored it decreases by at least $1$ (Claim~\ref{claim:point-scored-S-decreases}), and the game ends whenever the progress measure $S$ is $0$ (Claim~\ref{claim:progmeasure0endgame}), the Prover can ensure that the number of points scored by the Delayer is at most $(n+2)/3$. From Claim~\ref{claim: rank equals Prover Delayer game value}, we infer that $\rank(F) \leq (n+2)/3$, completing the proof.
    \end{proof}

    \begin{proof}[Proof of Theorem~\ref{thm: rank of complete nand tree}]
        Combining Claim~\ref{claim:rank-lb} and Claim~\ref{claim:rank-ub} completes the proof.
    \end{proof}

    \begin{proof}[Proof of Theorem~\ref{thm: rank rrank separation nand tree}]
        Since randomized rank is at most randomized decision tree complexity, Theorem~\ref{thm: saks wigderson} implies
        \[
        \rrank(F) \leq \mathsf{R}(f) = \Theta\rbra{n^{\log\frac{1 + \sqrt{33}}{4}}} \approx \Theta\rbra{n^{.753\dots}}.
        \]
        Theorem~\ref{thm: rank of complete nand tree} implies
        \[
        \rank(F) = \frac{n+2}{3}.
        \]
        This proves the theorem.
    \end{proof}

    \section{Proof of Theorem~\ref{thm: weights in program give query upper bound}}\label{sec: bt span program}

    The main results of this section provide a characterization of the optimal witness complexity and objective value of the dual adversary bound, based on the weighting scheme in Beigi and Taghavi's construction of an NBSPwOI and a dual adversary solution for $\leaff$ given a relation $f \subseteq \zone^n \times \calR$ and a decision tree $\T$ computing it~\cite[Section~3]{BT20} (recall from Section~\ref{sec: prelims} that $\leaff$ takes an input $x \in \zone^n$ and outputs the leaf of $\T$ reached on input $x$), in terms of the objective value of Program~\ref{program: weights}. We describe their construction in a modular fashion: we leave the choice of `weights' of the vectors in the span program and dual adversary solution unfixed. We show that the witness complexity and dual adversary bounds thus obtained are captured by the objective value of Program~\ref{program: weights}. In the next section we demonstrate a choice of weights and prove its optimality. In Appendix~\ref{app: another weight sqrt llogl} we exhibit another interesting choice of weights.

    \subsection{Span Program Construction}\label{sec: construction}

    In order to define the NBSPwOI, we first assign strictly positive real weights $W_e$ to all edges $e$ in the decision tree. These weights play a crucial role in the witness complexity analysis.

    The following is the NBSPwOI for $\leaff$.

    \begin{itemize}
        \item The vector space is the span of $\cbra{\ket{v} : v \in V(\T)}$.
        \item The input vectors are
        \begin{equation}\label{eqn: weights in input vectors}
            I_{j, q} = \bigcup_{v \in V_I(\T) : J(v) = j} \cbra{\sqrt{W_{e(v, N(v,q))}}\rbra{\ket{v} - \ket{N(v, q)}}}.
        \end{equation}
        That is, for all $j \in [n]$ and $q \in \zone$, the input vectors correspond to edges corresponding to answers of queries of the form $x_j = q$. In other words, for every vertex $v \in V_I(\T)$, $e(v,N(v,x_{J(v)}))$ is always the unique available outgoing edge of $v$. Moreover, these vectors are weighted, and we leave these weights variable for now.
        \item Let $r$ denote the root vertex of $\T$. For each leaf $u$ of $\T$, the associated target vector is given by
        \[
        \ket{t_u} = \ket{r} - \ket{u}.
        \]
    \end{itemize}
    We now give positive and negative witnesses for every $x \in D_f$, argue that the above span program evaluates $\leaff$, and analyze the positive and negative witness complexities.

    Note that, we use $v \in P_x$ to denote a vertex in the path $P_x$, and, we use $e \in P_x$ to denote an edge in the path $P_x$. For every $x \in D_f$, we can express the corresponding target vector by a telescoping sum of vectors that are all available to $x$, as
    \begin{align*}
        \ket{t_{\leaff(x)}} &= \ket{r} - \ket{\leaff(x)} = \sum_{v \in P_x \setminus \{\leaff(x)\}} \ket{v} - \ket{N(v,x_{J(v)})} \\
        &= \sum_{v \in P_x \setminus \{\leaff(x)\}}\frac{1}{\sqrt{W_{e(v, N(v, x_{J(v)}))}}}\rbra{\sqrt{W_{e(v, N(v, x_{J(v)}))}}\rbra{\ket{v} - \ket{N(v, x_{J(v)})}}}.
    \end{align*}
    Note from Equation~\eqref{eqn: weights in input vectors} that all the vectors in the curly braces above are available for $x$. Any vector in $\C^{I}$ can be viewed as indexed by elements of $I$. Define $\ket{w_x} \in \C^{|I|}$ as follows: For a node $v \in P_x \setminus \cbra{\leaff(x)}$, assign the value $\frac{1}{\sqrt{W_{e(v, N(v, x_{J(v)}))}}}$ to the entry corresponding to the vector $\sqrt{W_{e(v, N(v, x_{J(v)}))}}\rbra{\ket{v} - \ket{N(v, x_{J(v)})}}$, and set all other coefficients to 0.
    On the other hand, we let
    \[\ket{\overline{w}_x} = \sum_{v \in P_x} \ket{v}.\]
    For any vector $\ket{v''}$ in $I(x)$, because it is of the form
    \[\ket{v''}=\sqrt{W_{e(v',N(v',x_{J(v')}))}} \left(\ket{v'} - \ket{N(v',e(v',x_{J(v')}))}\right)\]
    for some $v' \in V_I(\T)$, we have
    \[\braket{\overline{w}_x}{v''} = \sum_{v \in P_x} \sqrt{W_{e(v',N(v',x_{J(v')}))}} \left(\braket{v}{v'} - \braket{v}{N(v',e(v',x_{J(v')}))}\right) = 0.\]
    For a leaf $u \neq \leaff(x)$ of $\T$,
    \begin{equation*}
        \braket{t_u}{\overline{w}_x} = \sum_{v \in P_x} \braket{r}{v} - \sum_{v \in P_x} \braket{u}{v} = 1 - 0 = 1.
    \end{equation*}
    This implies that the NBSPwOI indeed computes $\leaff$. For the positive and negative witness sizes, we have
    \[\wsizePos(P,w,\overline{w}) = \max_{x \in D_f}\norm{\ket{w_x}}^2 = \max_{x \in D_f} \sum_{v \in P_x \setminus \{\leaff(x)\}} \frac{1}{W_{e(v,N(v,x_{J(v)}))}} = \max_{P \in P(\T)} \sum_{e \in P} \frac{1}{W_e},\]
    and
    \begin{align*}
        \wsizeNeg(P,w,\overline{w}) &= \max_{x \in D_f} \norm{A^{\dagger} \ket{\overline{w}_x}}^2 = \max_{x \in D_f} \norm{\sqrt{W_{e(v,N(v,x_{J(v)}))}}\left(\bra{v} - \bra{N(v,x_{J(v)})}\right)\ket{\overline{w}_x}}^2 \\
        &= \sum_{v \in P_x \setminus \{\leaff(x)\}} W_{e(v,N(v,\overline{x_{J(v)}}))} = \max_{P \in P(\T)} \sum_{e \in \overline{P}} W_e.
    \end{align*}
    Now, it remains to find the weight assignment $W$ that minimizes the total complexity of the NBSPwOI, which is given by
    \[\wsize(P,w,\overline{w}) = \sqrt{\wsizeNeg(P,w,\overline{w}) \cdot \wsizePos(P,w,\overline{w})} = \sqrt{\max_{P \in P(\T)} \sum_{e \in P} \frac{1}{W_e} \cdot \max_{P \in P(\T)} \sum_{e \in \overline{P}} W_e}.\]

    Thus, if we assign weights $W$ to the edges of $\T$ as earlier in this section, and set $\alpha = \max_{P \in P(\T)} \sum_{e \in \overline{P}} W_e$ and $\beta = \max_{P \in P(\T)} \sum_{e \in P} \frac{1}{W_e}$, the construction from earlier in this section gives rise to an explicit NBSPwOI computing $\leaff$
    with witness complexity of $\sqrt{\alpha\beta}$. This is exactly captured by Program~\ref{program: weights}, giving us the following theorem.

    \begin{theorem}\label{thm: witness size bounded by opt}
        Let $f \subseteq \zone^n \times \calR$ be a relation, let $\T$ be a decision tree computing $f$ and let $\OPT_\T$ denote the optimal value of Program~\ref{program: weights}. Then, for the construction of $(P, w, \overline{w})$ with variable weights as earlier in this section,
        \[
        \wsize(P, w, \overline{w}) = \OPT_\T.
        \]
    \end{theorem}

    In the next subsection we show that a solution to Program~\ref{program: weights} also gives a dual adversary solution with the same objective value.

    \subsection{Dual Adversary Solution}\label{sec: dual adversary solution}

    We give a refined analysis of Beigi and Taghavi's construction of a dual adversary solution for a relation $f\subseteq \{0,1\}^n \times \calR$ given a deterministic decision tree $\T$ that computes $f$. Recall that for a deterministic tree $\T$ computing $f$, $\leaff$ is the function that takes input $x \in D_f$ and outputs the leaf of $\T$ reached on input $x$. We show that a dual adversary solution for $\leaff$ can also be obtained by different settings of weights as in the previous subsection, and obtain a corresponding dual adversary bound as the optimum value of Program~\ref{program: weights}.

    We construct vectors $\cbra{\ket{u_{xj}} : x \in D_f, j \in [n]}$ and $\cbra{\ket{w_{xj}} : x \in D_f, j \in [n]}$ that are feasible solutions to Program~\ref{program: sdp dual f} for $\leaff$, which we recall below.

    \begin{table}[H]
        \centering~
        \begin{tabular}{|lll|}
            \hline
            Variables & $\cbra{\ket{u_{xj}} : x \in D_f, j \in [n]}$ and $\cbra{\ket{w_{xj}} : x \in D_f, j \in [n]}, d$ & \\
            Minimize  & $\max_{x \in D_f} \max\cbra{\sum_{j=1}^{n} \norm{\ket{u_{xj}}}^2, \sum_{j=1}^{n}\norm{\ket{w_{xj}}}^2} $ & \\
            s.t.
            & $\sum_{j \in [n] : x_j\neq y_j} \braket{u_{xj}}{w_{yj}} = 1- \delta_{\leaff(x),\leaff(y)}$ & $\forall x,y \in D_f$ \\
            & $\ket{u_{xj}}, \ket{w_{xj}} \in \mathbb{C}^{d}$ & for all $x \in D_f$ \\
            \hline
        \end{tabular}
        \caption{\label{program: sdp dual ftilde} Dual SDP for $\leaff$}
    \end{table}

    Let $V_I(\T)$ denote the set of internal nodes of $\T$. Consider the basis set $\{\ket{v}: v \in V_I(\T)\}$.
    We construct the vectors $\ket{u_{xj}}$ and $\ket{w_{xj}}$ in the space $\mathbb{C}^{V_I(\T)}$. Additionally, we use $V_j(\T)$ to denote the set of vertices associated with query index $j$, i.e., $V_j(\T) = \cbra{v \in V_I(\T) : J(v) = j}$.

    Define $\ket{u_{xj}}$ and $\ket{w_{xj}}$ as follows.
    \begin{equation}
        \label{eq:DetUxj}
        \ket{u_{xj}}=
        \begin{cases} \frac{1}{\sqrt{W_{e(v,N(v,x_j))}}}\ket{v} & \text{if } \exists v \in P_x \cap V_j(\T), \\
            0 & \text{otherwise},
        \end{cases}
    \end{equation}
    and,
    \begin{equation}
        \label{eq:DetWxj}
        \ket{w_{xj}}=
        \begin{cases} \sqrt{W_{e(v,N(v,\overline{x_j}))}}\ket{v} & \text{if } \exists v \in P_x \cap V_j(\T), \\
            0 & \text{otherwise}.
        \end{cases}
    \end{equation}

    We claim that these vectors form a feasible solution to Program~\ref{program: sdp dual ftilde}. Fix $x,y \in D_f$ with $\leaff(x)\neq \leaff(y)$. We now verify that the corresponding equality constraint in Program~\ref{program: sdp dual ftilde} is satisfied.
    \begin{itemize}
        \item There is a unique vertex in $\T$ where $P_x$ and $P_y$ deviate. Let $v \in V_I(\T)$ denote this vertex, and let its associated query index be $J(v) = i$. We have $v \in P_x \cap P_y$ and $x_i\neq y_i$. In that case $\braket{u_{xi}}{w_{yi}}=1$ by the definitions of $\ket{u_{xi}}$ and $\ket{w_{yi}}$ from Equations~\eqref{eq:DetUxj} and~\eqref{eq:DetWxj}.
        \item Consider an index $j \in [n] \setminus \cbra{i}$ such that $x_j \neq y_j$. Let $v'$ and $v''$ be the vertices on $P_x$ and $P_y$, respectively (if they exist), with $J(v') = J(v'') = j$. By the previous point, $v' \notin P_y$ and $v'' \notin P_x$. Thus, $\braket{v'}{v''} = 0$ from Equations~\eqref{eq:DetUxj} and~\eqref{eq:DetWxj}, which implies $\braket{u_{xj}}{w_{yj}} = 0$.
    \end{itemize}
    Thus, we have for all $x, y \in D_f$ with $\leaff(x)\neq \leaff(y)$,
    \[
    \sum_{j \in [n] : x_j \neq y_j} \braket{u_{xj}}{w_{yj}}=1-\delta_{\leaff(x),\leaff(y)}.
    \]

    In the case when $\leaff(x) = \leaff(y)$, the right hand side in the constraint evaluates to $0$ and so does the left side, because the indices where $x$ and $y$ differ cannot be queried on their path since $x$ and $y$ reach the same leaf in $\T$. Therefore, the set of vectors $\cbra{\ket{u_{xj}} : x \in D_f, j \in [n]}$ and $\cbra{\ket{w_{xj}} : x \in D_f, j \in [n]}$, and $d = |V_I(\T)|$ form a feasible solution to Program~\ref{program: sdp dual ftilde} for $\leaff$.

    \begin{theorem}\label{thm: dadv is bounded by opt}
        Let $f \subseteq \zone^n \times \calR$ be a relation, and let $\T$ be a decision tree computing it. Let $C$ denote the optimal value of Program~\ref{program: sdp dual ftilde} with variable weights as earlier in this section, and let $\OPT_\T$ denote the optimal value of Program~\ref{program: weights}. Then $C = \OPT_\T$.
    \end{theorem}

    \begin{proof}[Proof of Theorem~\ref{thm: dadv is bounded by opt}]
        Let $C$ denote the objective value of Program~\ref{program: sdp dual ftilde} with the settings of vectors $\cbra{\ket{u_{xj}} : x \in D_f, j \in [n]}$ and $\cbra{\ket{w_{xj}} : x \in D_f, j \in [n]}$ as defined in Equations~\eqref{eq:DetUxj} and~\eqref{eq:DetWxj}, and $d = |V_I(\T)|$.

        We now argue that $C=\OPT_{\T}$, where $\OPT_{\T}$ denotes the optimal solution of Program~\ref{program: weights}. First note that for all $x \in D_f$,
        \begin{equation*}
            \sum_{j=1}^{n} \norm{\ket{u_{xj}}}^2 = \sum_{e \in P_x} \frac{1}{W_e} \qquad \text{and} \qquad \sum_{j=1}^{n} \norm{\ket{w_{xj}}}^2 = \sum_{e \in \overline{P}_x} W_e.\\
        \end{equation*}
        Thus,
        \begin{align*}
            C &= \min \max_{x \in D_f} \max  \cbra{\sum_{j=1}^{n} \norm{\ket{u_{xj}}}^2,  \sum_{j=1}^{n} \norm{\ket{w_{xj}}}^2} = \min \max_{x \in D_f} \max \cbra{\sum_{e \in P_x} \frac{1}{W_e}, \sum_{e \in \overline{P}_x} W_e}.
        \end{align*}
        Thus we can alternatively view $C$ to be an optimal solution to the following optimization program.

        \begin{table}[H]
            \centering~
            \begin{tabular}{|llllll|}
                \hline
                Variables & $\cbra{W_e: e~\textnormal{is an edge in}~\T}, \alpha, \beta$ & & & & \\
                Minimize  & $\max\{\alpha, \beta\}$ & & & & \\
                s.t. & $\sum_{e \in \overline{P}}W_e$ & $\leq \alpha,$ & & for all paths $P \in P(\T)$ & \\
                & $\sum_{e \in P}\frac{1}{W_e}$ & $\leq \beta,$ & & for all paths $P \in P(\T)$ & \\
                & $W_e$ & $> 0,$ & & for all edges $e$ in $\T$ & \\
                & $\alpha,\beta$ & $\geq 0.$ & & & \\
                \hline
            \end{tabular}
            \caption{\label{program: sdp dual ftilde program}}
        \end{table}

        We now show $C = \OPT_\T$. Let $\cbra{W_e : e~\textnormal{edge in}~\T}, \alpha, \beta$ be settings of variables in a feasible solution to Program~\ref{program: sdp dual ftilde program}, with objective value $C$. Clearly the same settings of variables also form a feasible solution to Program~\ref{program: weights}, since the constraints are the same. The corresponding objective value of Program~\ref{program: weights} is $\sqrt{\alpha \beta} \leq \max\cbra{\alpha, \beta} = C$. Thus, $\OPT_\T \leq C$.

        In the other direction, let $\cbra{W_e : e~\textnormal{edge in}~\T}, \alpha, \beta$ be settings of variables in a feasible solution to Program~\ref{program: weights} with objective value $\OPT_\T$. Set
        \begin{align*}
            W'_e & = \sqrt{\beta/\alpha} \cdot W_e \qquad \textnormal{for all edges $e$ in $\T$},\\
            \alpha' & = \sqrt{\alpha \beta},\\
            \beta' & = \sqrt{\alpha\beta}.
        \end{align*}
        It is easy to verify that this setting of variables is feasible for Program~\ref{program: sdp dual ftilde program}, and attains objective value $\max\cbra{\alpha', \beta'} = \sqrt{\alpha\beta} = \OPT_\T$. Thus, $C \leq \OPT_\T$, proving the theorem.
    \end{proof}

    We now prove Theorem~\ref{thm: weights in program give query upper bound}, using results proved earlier in this section.

    \begin{proof}[Proof of Theorem~\ref{thm: weights in program give query upper bound}]
        We give two proofs, one via span program witness complexity, and another via the dual adversary bound.
        \begin{enumerate}
            \item Consider the NBSPwOI $(P, w, \overline{w})$ for $\leaff$ as constructed in Section~\ref{sec: construction}. Theorem~\ref{thm: witness size bounded by opt} implies $\wsize(P, w, \overline{w}) \leq \OPT_\T$. Theorem~\ref{thm: bt19} implies $\Q(\leaff) = O(\OPT_\T)$.
            \item Consider the dual adversary solution for $\leaff$ as constructed in Section~\ref{sec: dual adversary solution}. Theorems~\ref{thm:DualSDP-QuantumQueryUpperBound} and~\ref{thm: dadv is bounded by opt} imply $\Q(\leaff) = O(C) = O(\OPT_\T)$.
        \end{enumerate}
    \end{proof}

    \section{An Optimal Weight Assignment}
    \label{sec: weight assignments}

    It now remains to investigate how we can assign weights to edges in a decision tree so as to optimize the objective value of Program~\ref{program: weights}. Beigi and Taghavi gave an explicit weighting scheme by coloring the edges of the decision tree with two colors, and then assigning weights depending on the color that the edge is colored by~\cite[Section~3]{BT20}. They raised the question whether their scheme can be significantly improved upon. We answer this question affirmatively, by giving an \emph{optimal} solution to the weight optimization program from Definition~\ref{def:weight-program}, and providing a constructive algorithm to compute the optimal weights in this section. We also give an alternative, albeit sub-optimal, assignment of weights in Appendix~\ref{app: another weight sqrt llogl}.


    First, we define the weight assignment, which we will refer to as the \textit{canonical weight assignment}, in Definition~\ref{def:weight-assignment}. The construction we present is recursive, in the sense that we first assign weights to the edges connected to the leaves, and subsequently work our way up the tree until we reach the root node. Then, we prove its optimality, resulting in Theorem~\ref{thm:weight-optimality}. Finally, we prove Corollary~\ref{cor:ubs}, which gives some upper bounds on the optimal objective value, in terms of natural measures of the decision tree.

    \begin{definition}[Canonical weight assignment]
        \label{def:weight-assignment}
        Let $\T$ be a non-trivial decision tree with root node $r$. Let $L$ and $R$ be the two children nodes of $r$, connected to $r$ by the edges $e_L$ and $e_R$, respectively. Let $\T_L$ and $\T_R$ be the subtrees of $\T$ rooted at $L$ and $R$, respectively. Then, we assign weights to $e_L$ and $e_R$, by setting
        \begin{align*}
            W_{e_L} &= \frac{\OPT_{\T_L} - \OPT_{\T_R} + \sqrt{(\OPT_{\T_L} - \OPT_{\T_R})^2 + 4}}{2}, \\
            W_{e_R} &= \frac{\OPT_{\T_R} - \OPT_{\T_L} + \sqrt{(\OPT_{\T_R} - \OPT_{\T_L})^2 + 4}}{2}.
        \end{align*}
    \end{definition}

    In order to define the weights $W_{e_L}$ and $W_{e_R}$, we need to know the optimal values of Program~\ref{program: weights} for the subtrees $\T_L$ and $\T_R$. We now proceed to show how one can compute these optimal values via a recurrence relation.

    \begin{lemma}
        \label{lem:complexity-ub}
        Let $\T$ be a non-trivial decision tree, and let $L$ and $R$ be the two children nodes of the root node. Let $\T_L$ and $\T_R$ be the subtrees of $\T$ rooted at $L$ and $R$, respectively. Then,
        \[\OPT_\T \leq \frac{\OPT_{\T_L} + \OPT_{\T_R} + \sqrt{(\OPT_{\T_L} - \OPT_{\T_R})^2 + 4}}{2}.\]
    \end{lemma}

    \begin{proof}
        Suppose we have weight assignments $W_L$ and $W_R$ on the subtrees $\T_L$ and $\T_R$, and positive parameters $\alpha_L, \alpha_R$ and $\beta_L, \beta_R$ such that they form feasible solutions to the optimization programs for the left and right subtrees $\T_L$ and $\T_R$, and such that they attain the optimal values $\sqrt{\alpha_L\beta_L} = \OPT_{\T_L}$ and $\sqrt{\alpha_R\beta_R} = \OPT_{\T_R}$. Now, let the weighting scheme $W$ on $\T$ be defined such that
        \begin{equation}\label{eqn: wl}
            W_{e} = \sqrt{\beta_L/\alpha_L}{(W_L)}_e
        \end{equation}
        for all edges $e$ in $\T_L$,
        \begin{equation}\label{eqn: wr}
            W_{e'} = \sqrt{\beta_R/\alpha_R}{(W_R)}_{e'}
        \end{equation}
        for all edges $e'$ in $\T_R$, and choose $W_{e_L}$ and $W_{e_R}$ as in Definition~\ref{def:weight-assignment}. Furthermore, let
        \begin{align}\label{eqn: settings of alpha beta}
            \alpha &:= \max\left\{\OPT_{\T_L} + W_{e_R}, \OPT_{\T_R} + W_{e_L}\right\}, \\
            \beta &:= \max\left\{\OPT_{\T_L} + \frac{1}{W_{e_L}}, \OPT_{\T_R} + \frac{1}{W_{e_R}}\right\}.
        \end{align}
        First observe that for a path $P \in P(\T)$ whose first edge is $e_L$ ($e_R$, respectively) must have all remaining edges in $\T_L$ ($\T_R$, respectively).
        For every path $P \in P(\T)$ containing $e_L$ (essentially the same calculation also shows the same upper bound for paths containing $e_R$), we have
        \begin{align*}
            \sum_{e \in P} \frac{1}{W_e} & = \frac{1}{W_{e_L}} + \sqrt{\frac{\alpha_L}{\beta_L}} \sum_{e \in P \setminus \{e_L\}} \frac{1}{(W_L)_e} \tag*{by Equation~\eqref{eqn: wl}}\\
            & \leq \frac{1}{W_{e_L}} + \sqrt{\alpha_L\beta_L}\\
            & = \frac{1}{W_{e_L}} + \OPT_{\T_L} \leq \beta,
        \end{align*}
        where the last inequality and equality follow since $(W_L, \alpha_L, \beta_L)$ is a feasible solution to Program~\ref{program: weights} for $\T_L$.
        For every path $P \in P(\T)$ containing $e_L$ (essentially the same calculation also shows the same upper bound for paths containing $e_R$), we have
        \begin{align*}
            \sum_{e \in \overline{P}} W_e & = W_{e_R}  + \sqrt{\frac{\beta_L}{\alpha_L}} \sum_{e \in \overline{P} \setminus \{e_R\}} (W_L)_e \tag*{by Equation~\eqref{eqn: wr}}\\
            & \leq  W_{e_R} + \sqrt{\alpha_L\beta_L}\\
            & = W_{e_R} + \OPT_{\T_L} \leq \alpha,
        \end{align*}
        where the last inequality and equality follow since $(W_L, \alpha_L, \beta_L)$ is a feasible solution to Program~\ref{program: weights} for $\T_L$.	Thus all constraints of Program~\ref{program: weights} for $\T$ are satisfied with these values of $\alpha, \beta$ and the weighting scheme $W$. Hence,
        \begin{align}\label{eqn: opt upper bound sqrt alpha beta}
            &\OPT_\T \leq \sqrt{\alpha\beta} \nonumber \\
            &\;\; = \sqrt{\max\left\{\OPT_{\T_L} + W_{e_R}, \OPT_{\T_R} + W_{e_L}\right\} \cdot \max\left\{\OPT_{\T_L} + \frac{1}{W_{e_L}}, \OPT_{\T_R} + \frac{1}{W_{e_R}}\right\}}.
        \end{align}
        Finally, observe from our choices of $W_{e_L}$ and $W_{e_R}$ from Definition~\ref{def:weight-assignment} that
        \begin{align*}
            \frac{1}{W_{e_L}} &= \frac{2}{\OPT_{\T_L} - \OPT_{\T_R} + \sqrt{(\OPT_{\T_L} - \OPT_{\T_R})^2 + 4}} \\
            &= \frac{2(\OPT_{\T_L} - \OPT_{\T_R} - \sqrt{(\OPT_{\T_L} - \OPT_{\T_R})^2 + 4})}{(\OPT_{\T_L} - \OPT_{\T_R})^2 - (\OPT_{\T_L} - \OPT_{\T_R})^2 - 4} \\
            &= \frac{\OPT_{\T_R} - \OPT_{\T_L} + \sqrt{(\OPT_{\T_L} - \OPT_{\T_R})^2 + 4}}{2} = W_{e_R}.
        \end{align*}
        Hence Equation~\eqref{eqn: opt upper bound sqrt alpha beta} implies
        \begin{align*}
            \OPT_{\T} &\leq \max\left\{\OPT_{\T_L} + W_{e_R}, \OPT_{\T_R} + W_{e_L}\right\} \\
            &= \frac{\OPT_{\T_L} + \OPT_{\T_R} + \sqrt{(\OPT_{\T_L} - \OPT_{\T_R})^2 + 4}}{2},
        \end{align*}
        where the last equality follows from our choices of $W_{e_L}$ and $W_{e_R}$ from Definition~\ref{def:weight-assignment}.
        This completes the proof.
    \end{proof}

    We now show that this weight assignment is optimal.

    \begin{lemma}
        \label{lem:complexity-lb}
        Let $\T$ be a non-trivial decision tree, and let $L$ and $R$ be the two children nodes of the root node. Let $\T_L$ and $\T_R$ be the subtrees of $\T$ rooted at $L$ and $R$, respectively. Then,
        \[\OPT_\T \geq \frac{\OPT_{\T_L} + \OPT_{\T_R} + \sqrt{(\OPT_{\T_L} - \OPT_{\T_R})^2 + 4}}{2}.\]
    \end{lemma}

    \begin{proof}
        Suppose we have a weight assignment $W$ and variables $\alpha, \beta > 0$, such that $W$, $\alpha$ and $\beta$ form a feasible solution to the weight optimization program for $\T$, and attain optimality. We define
        \begin{align*}
            \alpha_L &= \max_{P \in P(\T_L)} \sum_{e \in \overline{P}} W_e, & \beta_L &= \max_{P \in P(\T_L)} \sum_{e \in P} \frac{1}{W_e}, \\
            \alpha_R &= \max_{P \in P(\T_R)} \sum_{e \in \overline{P}} W_e, & \beta_R &= \max_{P \in P(\T_R)} \sum_{e \in P} \frac{1}{W_e}.
        \end{align*}
        Observe that $W|_{\T_L}$, $\alpha_L$ and $\beta_L$ give a feasible solution to Program~\ref{program: weights} for $\T_L$. Similarly, $W|_{\T_R}$, $\alpha_R$ and $\beta_R$ give a feasible solution to Program~\ref{program: weights} for $\T_R$. This implies that
        \begin{equation}\label{eqn: optl optr alphalr betalr (payer. dutch joke hehehe)}
            \sqrt{\alpha_L\beta_L} \geq \OPT_{\T_L}, \qquad \sqrt{\alpha_R\beta_R} \geq \OPT_{\T_R}.
        \end{equation}
        Furthermore, we have
        \begin{align}
            \alpha &= \max_{P \in P(\T)} \sum_{e \in \overline{P}} W_e = \max\left\{\max_{P \in P(\T_L)} \sum_{e \in \overline{P}} W_e + W_{e_R}, \max_{P \in P(\T_R)} \sum_{e \in \overline{P}} W_e + W_{e_L}\right\}\nonumber \\
            &= \max\left\{\alpha_L + W_{e_R}, \alpha_R + W_{e_L}\right\}, \label{eqn: alpha opt quest}
        \end{align}
        and similarly
        \begin{align}
            \beta &= \max_{P \in P(\T)} \sum_{e \in P} \frac{1}{W_e} = \max\left\{\max_{P \in P(\T_L)} \sum_{e \in P} \frac{1}{W_e} + \frac{1}{W_{e_L}}, \max_{P \in P(\T_R)} \sum_{e \in P} \frac{1}{W_e} + \frac{1}{W_{e_R}}\right\}\nonumber \\
            &= \max\left\{\beta_L + \frac{1}{W_{e_L}}, \beta_R + \frac{1}{W_{e_R}}\right\}. \label{eqn: beta opt quest}
        \end{align}
        Finally, let $\gamma = \sqrt{\beta/\alpha}$, and observe that
        \begin{equation}\label{eqn: opt alpha beta gamma}
            \OPT_\T = \sqrt{\alpha\beta} = \frac{\beta}{\gamma} = \gamma\alpha = \frac12\left[\frac{\beta}{\gamma} + \gamma\alpha\right].
        \end{equation}
        Now, let
        \begin{equation}\label{eqn: delta value}
            \delta = \frac12 + \frac{\OPT_{\T_L} - \OPT_{\T_R}}{2\sqrt{(\OPT_{\T_L} - \OPT_{\T_R})^2 + 4}},
        \end{equation}
        and observe that $\delta \in [0,1]$. Thus,
        \begin{align*}
            \OPT_{\T} &= \frac12\left[\frac{\beta}{\gamma} + \gamma\alpha\right] \tag*{by Equation~\eqref{eqn: opt alpha beta gamma}}\\
            &= \frac12\left[\max\left\{\frac{\beta_L}{\gamma} + \frac{1}{\gamma W_{e_L}}, \frac{\beta_R}{\gamma} + \frac{1}{\gamma W_{e_R}}\right\} + \max\left\{\gamma\alpha_L + \gamma W_{e_R}, \gamma\alpha_R + \gamma W_{e_L}\right\}\right] \tag*{by Equations~\eqref{eqn: alpha opt quest} and~\eqref{eqn: beta opt quest}}\\
            &\geq \frac12\left[\delta\left(\frac{\beta_L}{\gamma} + \frac{1}{\gamma W_{e_L}}\right) + (1-\delta)\left(\frac{\beta_R}{\gamma} + \frac{1}{\gamma W_{e_R}}\right)\right. \\
            &\qquad\qquad \left.+ \delta\left(\gamma\alpha_L + \gamma W_{e_R}\right) + (1-\delta)\left(\gamma\alpha_R + \gamma W_{e_L}\right)\right] \tag*{since $\max\{a,b\} \geq \delta a + (1-\delta)b$ as $\delta \in [0, 1]$}\\
            &= \frac12\left[\delta\left(\frac{\beta_L}{\gamma} + \gamma\alpha_L\right) + (1-\delta)\left(\frac{\beta_R}{\gamma} + \gamma\alpha_R\right)\right. \\
            &\qquad\qquad \left.+ \left(\frac{\delta}{\gamma W_{e_L}} + (1-\delta)\gamma W_{e_L}\right) + \left(\delta\gamma W_{e_R} + \frac{1-\delta}{\gamma W_{e_R}}\right)\right] \tag*{rearranging terms}\\
            &\geq \delta\sqrt{\alpha_L\beta_L} + (1-\delta)\sqrt{\alpha_R\beta_R} + \sqrt{\delta(1-\delta)} + \sqrt{\delta(1-\delta)} \tag*{since $(a+b)/2 \geq \sqrt{ab}$ for all $a,b \geq 0$}\\
            &\geq \delta\OPT_{\T_L} + (1-\delta)\OPT_{\T_R} + 2\sqrt{\delta(1-\delta)} \tag*{by Equation~\eqref{eqn: optl optr alphalr betalr (payer. dutch joke hehehe)}} \\
            &= \frac{\OPT_{\T_L} + \OPT_{\T_R}}{2} + \frac{(\OPT_{\T_L} - \OPT_{\T_R})^2}{2\sqrt{(\OPT_{\T_L} - \OPT_{\T_R})^2 + 4}} \\
            &\qquad\qquad + 2\sqrt{\frac14 - \frac{(\OPT_{\T_L} - \OPT_{\T_R})^2}{4((\OPT_{\T_L} - \OPT_{\T_R})^2 + 4)}} \tag*{plugging the value of $\delta$ from Equation~\eqref{eqn: delta value}}\\
            &= \frac{\OPT_{\T_L} + \OPT_{\T_R}}{2} + \frac{(\OPT_{\T_L} - \OPT_{\T_R})^2}{2\sqrt{(\OPT_{\T_L} - \OPT_{\T_R})^2 + 4}} + \sqrt{\frac{4}{(\OPT_{\T_L} - \OPT_{\T_R})^2 + 4}}\\
            &= \frac{\OPT_{\T_L} + \OPT_{\T_R}}{2} + \frac{(\OPT_{\T_L} - \OPT_{\T_R})^2 + 4}{2\sqrt{(\OPT_{\T_L} - \OPT_{\T_R})^2 + 4}} \\
            &= \frac{\OPT_{\T_L} + \OPT_{\T_R} + \sqrt{(\OPT_{\T_L} - \OPT_{\T_R})^2 + 4}}{2},
        \end{align*}
        completing the proof.
    \end{proof}

    We can now combine both lemmas above into the following theorem, providing a recursive characterization of the optimal value of Program~\ref{program: weights} for all decision trees.
    \begin{theorem}
        \label{thm:weight-optimality}
        Let $\T$ be a non-trivial decision tree, and let $L$ and $R$ be the two children nodes of the root node. Let $\T_L$ and $\T_R$ be the subtrees of $\T$ rooted at $L$ and $R$, connected to the root node by edges $e_L$ and $e_R$, respectively. Then,
        \[
        \OPT_\T = \frac{\OPT_{\T_L} + \OPT_{\T_R} + \sqrt{(\OPT_{\T_L} - \OPT_{\T_R})^2 + 4}}{2}.
        \]
    \end{theorem}

    \begin{proof}
        The upper bound on follows from Lemma~\ref{lem:complexity-ub} and the lower bound follows from Lemma~\ref{lem:complexity-lb}.
    \end{proof}

    Observe that Definition~\ref{def:weight-assignment} can be used to recursively assign optimal weights to the edges of any given decision tree. We display this technique in two examples in Figure~\ref{fig:example-weight-assignment}. Note that the objective value of Program~\ref{program: weights} for a trivial decision tree (a single node) is 0, which provides the basis for our recursion.

    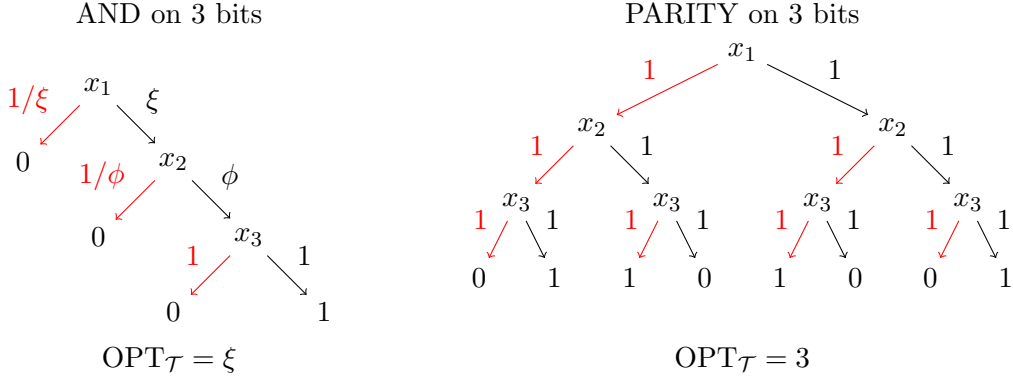
\begin{figure}[h!]
        \centering
        \begin{tabular}{ccc}
            AND on $3$ bits && PARITY on $3$ bits \\
            \begin{tikzpicture}[yscale=-1]
                \node (x1) at (0,0) {$x_1$};
                \node (x2) at (1,1) {$x_2$};
                \node (x3) at (2,2) {$x_3$};
                \node (r1) at (-1,1) {$0$};
                \node (r2) at (0,2) {$0$};
                \node (r3) at (1,3) {$0$};
                \node (a) at (3,3) {$1$};
                \draw[->] (x1) to node[above right] {$\xi$} (x2);
                \draw[->] (x2) to node[above right] {$\phi$} (x3);
                \draw[->,red] (x1) to node[above left] {$1/\xi$} (r1);
                \draw[->,red] (x2) to node[above left] {$1/\phi$} (r2);
                \draw[->,red] (x3) to node[above left] {$1$} (r3);
                \draw[->] (x3) to node[above right] {$1$} (a);
            \end{tikzpicture} & \hspace{2em} & \raisebox{1.2em}{\begin{tikzpicture}[yscale=-1]
                    \node (x1) at (0,0) {$x_1$};
                    \node (x2) at (2,1) {$x_2$};
                    \node (y2) at (-2,1) {$x_2$};
                    \node (a1) at (-3,2) {$x_3$};
                    \node (r1) at (-1,2) {$x_3$};
                    \node (r2) at (1,2) {$x_3$};
                    \node (a2) at (3,2) {$x_3$};

                    \node (b000) at (-3.5,3) {$0$};
                    \node (b001) at (-2.5,3) {$1$};
                    \node (b010) at (-1.5,3) {$1$};
                    \node (b011) at (-.5,3) {$0$};
                    \node (b100) at (.5,3) {$1$};
                    \node (b101) at (1.5,3) {$0$};
                    \node (b110) at (2.5,3) {$0$};
                    \node (b111) at (3.5,3) {$1$};
                    \draw[->] (x1) to node[above right] {$1$} (x2);
                    \draw[->,red] (x1) to node[above left] {$1$} (y2);
                    \draw[->] (x2) to node[above right] {$1$} (a2);
                    \draw[->,red] (x2) to node[above left] {$1$} (r2);
                    \draw[->] (y2) to node[above right] {$1$} (r1);
                    \draw[->,red] (y2) to node[above left] {$1$} (a1);

                    \draw[->,red] (a1) to node[above left] {$1$} (b000);
                    \draw[->] (a1) to node[above right] {$1$} (b001);
                    \draw[->,red] (r1) to node[above left] {$1$} (b010);
                    \draw[->] (r1) to node[above right] {$1$} (b011);
                    \draw[->,red] (r2) to node[above left] {$1$} (b100);
                    \draw[->] (r2) to node[above right] {$1$} (b101);
                    \draw[->,red] (a2) to node[above left] {$1$} (b110);
                    \draw[->] (a2) to node[above right] {$1$} (b111);

            \end{tikzpicture}} \\
            $\OPT_{\T} = \xi$ && $\OPT_{\T} = 3$
        \end{tabular}
        \caption{Examples of optimal weight assignments for two different decision trees. The red and black edges indicate the edges taken when the output of the query is $0$ and $1$, respectively, and the edge labels represent the weights. Left: Canonical weight assignment of the decision tree for the AND function on $3$ bits, where $\phi = \frac{1 + \sqrt{5}}{2}$, and $\xi = \frac{\phi+\sqrt{\phi+5}}{2}$. The objective value is $\xi$. Right: Canonical weight assignment of the decision tree for PARITY on $3$ bits, with optimal value $3$.}
        \label{fig:example-weight-assignment}
    \end{figure}

    In Corollary~\ref{cor:ubs} we exhibit two ways in which we can conveniently upper bound the optimum value of Program~\ref{program: weights} in terms of well-studied measures of the underlying decision tree.

    \begin{proof}[Proof of Corollary~\ref{cor:ubs}]
        Theorem~\ref{thm: weights in program give query upper bound} implies that it suffices to prove both of the required upper bounds on $\OPT_\T$ rather than $\Q(f)$.
        The rank-depth bound follows directly from the bound $\OPT_\T \leq 2\sqrt{G(\T) \cdot \textnormal{depth}(T)}$ derived in \cite[Theorem~2]{BT20}, and the equality $G(\T) = \rank(\T)$ from Claim~\ref{claim: gcoloring equals rank}.

        For proving the size bound, we use induction to show that $\OPT_\T \leq \sqrt{2\dtsize(\T)}$. First observe that it is true for the trivial decision tree. Next, suppose that it is true for the left and right subtrees $\T_L$ and $\T_R$ of $\T$. That is,
        \[
        \OPT_{\T_L} \leq \sqrt{2\dtsize(\T_L)}, \qquad \text{and} \qquad \OPT_{\T_R} \leq \sqrt{2\dtsize(\T_R)}.
        \]
        Then, by Theorem~\ref{thm:weight-optimality}, the square of the optimal value of Program~\ref{program: weights} for $\T$ equals
        \begin{align*}
            \OPT_\T^2 &= \frac{(\OPT_{\T_L} + \OPT_{\T_R})^2 + (\OPT_{\T_L} - \OPT_{\T_R})^2 + 4}{4} \\
            &\qquad+ \frac{2(\OPT_{\T_L} + \OPT_{\T_R})\sqrt{(\OPT_{\T_L} - \OPT_{\T_R})^2 + 4}}{4} \\
            &= \frac{2(\OPT_{\T_L}^2 + \OPT_{\T_R}^2) + 4}{4} + \frac{2(\OPT_{\T_L} + \OPT_{\T_R})\sqrt{(\OPT_{\T_L} - \OPT_{\T_R})^2 + 4}}{4}\\
            &\leq \frac{2(\OPT_{\T_L}^2 + \OPT_{\T_R}^2) + 4 + (\OPT_{\T_L} + \OPT_{\T_R})^2 + (\OPT_{\T_L} - \OPT_{\T_R})^2 + 4}{4} \tag*{since $2ab \leq a^2 + b^2$ for all $a,b \geq 0$}\\
            &= \OPT_{\T_L}^2 + \OPT_{\T_R}^2 + 2 \\
            &\leq 2\dtsize(\T_L) + 2\dtsize(\T_R) + 2 = 2\dtsize(\T).
        \end{align*}
        This completes the proof.
    \end{proof}

    Next, we note that the rank-depth and size bounds from Corollary~\ref{cor:ubs} are incomparable, as witnessed by the examples displayed in Figure~\ref{fig:bounds-separation}. In particular, our bounds are strictly stronger than those given by earlier works (Theorem~\ref{thm: linlin}) for the second tree in the figure.

    \begin{figure}[H]
        \centering
        \begin{tabular}{rlcrl}
            \multicolumn{2}{c}{Complete binary tree} & \hspace{4em} & \multicolumn{2}{c}{Balanced binary-AND tree} \\
            \multicolumn{2}{c}{\raisebox{1.6em}{\begin{tikzpicture}[yscale=-1]
                        \draw[pattern=north west lines, pattern color=gray] (0,0) to (1,1) to (-1,1) to cycle;
                        \draw[decorate, decoration={calligraphic brace,mirror,raise=3pt}] (-1,1) to node[below=4pt] {$n$ leaves} (1,1);
                        \draw[decorate, decoration={calligraphic brace,mirror,raise=3pt}] (-1,0) to node[left=4pt] {$\begin{array}{c}\text{depth}\\\log n\end{array}$} (-1,1);
            \end{tikzpicture}}} && \multicolumn{2}{c}{\begin{tikzpicture}[yscale=-1]
                    \draw (0,0) to (-.5,.5);
                    \draw (0,0) to (3.1,3.1);
                    \draw[pattern=north west lines, pattern color=gray] (-.5,.5) to (-1.5,1.5) to (.5,1.5) to cycle;
                    \draw[decorate, decoration={calligraphic brace,mirror,raise=5pt}] (-1.5,1.5) to node[below=4pt] {$n$ leaves} (.5,1.5);
                    \draw[decorate, decoration={calligraphic brace,raise=5pt}] (-1.5,1.5) to node[left=4pt] {$\begin{array}{c}\text{depth}\\\log n\end{array}$} (-1.5,.5);
                    \foreach \x in {0,.1,.2,...,3} {
                        \draw (\x,\x) to ({\x-.1},{\x+.1});
                    }
                    \draw[decorate, decoration={calligraphic brace,mirror,raise=7pt}] (.1,.1) to node[below left=-1em] {\rotatebox{-45}{$n$ leaves}} (3,3);
            \end{tikzpicture}} \\
            $\OPT_{\T}$\hspace{-.9em} & $= O(\log n)$ && $\OPT_{\T}$\hspace{-.9em} & $= O(\sqrt{n})$ \\
            $\sqrt{\rank(\T)\text{depth}(\T)}$\hspace{-.9em} & $= O(\log n)$ && $\sqrt{\rank(\T)\text{depth}(\T)}$\hspace{-.9em} & $= O(\sqrt{n\log n})$ \\
            $\sqrt{\dtsize(\T)}$\hspace{-.9em} & $= O(\sqrt{n})$ && $\sqrt{\dtsize(\T)}$\hspace{-.9em} & $= O(\sqrt{n})$
        \end{tabular}
        \caption{Examples showing separations between the two bounds derived in Corollary~\ref{cor:ubs}. The shaded regions represent complete binary trees. In the left example the rank-depth bound beats the size bound, whereas in the right example the opposite is true.}
        \label{fig:bounds-separation}
    \end{figure}
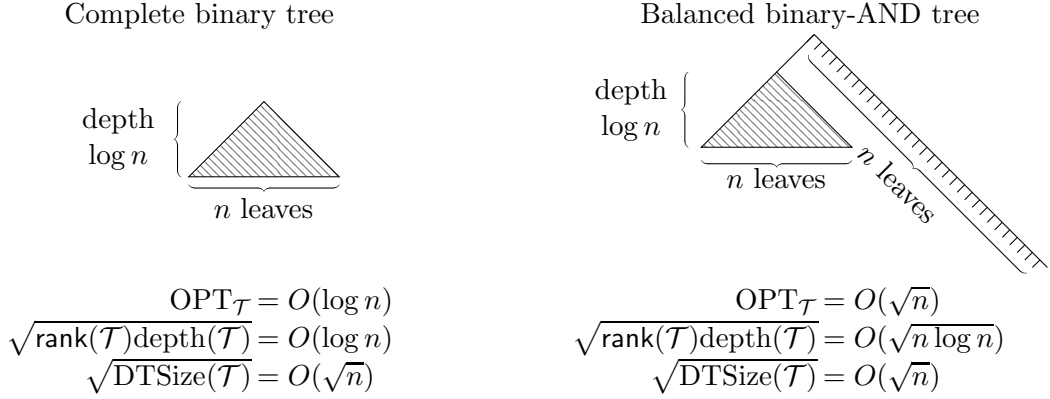

    Finally, we note that we can obtain analogous quantum query upper bounds to those in Corollary~\ref{cor:ubs} when the initial tree is a randomized one.

    \begin{corollary}\label{cor: quantum upper bound in terms of leaves}
        Let $f \subseteq \zone^n \times \calR$ be a relation. Then,
        \[
        \Q_{2/5}(f) = O(\sqrt{\rdtsize(f)}).
        \]
        Moreover, let $\T$ be a randomized decision tree computing $f$ with depth $T$ and randomized rank $G$. Then,
        \[
        \Q_{2/5}(f) = O\rbra{\sqrt{TG}}.
        \]
    \end{corollary}

    The proof of Corollary~\ref{cor: quantum upper bound in terms of leaves} follows along similar lines as the proof of Theorem~\ref{thm: qq upper bound in terms of randomized rank}, but we include it below for completeness.

    \begin{proof}
        Let $L = \rdtsize(f)$, and let $\T$ be a randomized decision tree witnessing this.
        For a decision tree $\T_\mu$ in the support of $\T$, define $f_\mu : \zone^n \to \calR$ be the function computed by $\T_\mu$.
        Corollary~\ref{cor:ubs} implies $\Q(\tilde{f_\mu}) = O\rbra{\sqrt{\dtsize(\T_\mu)}} = O\rbra{\sqrt{L}}$.
        By standard error-reduction techniques, $\Q_{9/10}(f_\mu) = O(\sqrt{L})$.

        A quantum query algorithm for $f$ is as follows: Sample $\T_\mu$ from the support of $\T$ according to its underlying distribution, and compute $f_\mu$ using the above-mentioned algorithm.
        The cost of this algorithm is $O(\sqrt{L})$ and the success probability is at least $(9/10)\cdot(2/3) = 3/5$ for all inputs $x \in D_f$, which proves the first query upper bound.

        A similar argument shows the second query upper bound as well.
    \end{proof}

    \paragraph*{Acknowledgements} We thank Ronald de Wolf and the anonymous FSTTCS reviewers for helpful comments.

    AC: This work was done while the author was at the Institute for Logic, Language, and Computation, University of Amsterdam and QuSoft. Supported by a Simons-CIQC postdoctoral fellowship through NSF QLCI Grant No.\ 2016245.

    NM: This work was done while the author was at QuSoft and CWI, Amsterdam, and supported by the Dutch Research Council (NWO/OCW), as part of the Quantum Software Consortium programme (project number 024.003.037), and through QuantERA ERA-NET Cofund project QuantAlgo (project number 680-91-034).

    SP: This work was done while the author was at QuSoft and CWI, Amsterdam, and supported by Robert Bosch Stiftung and by NWO Gravitation grants NETWORKS and QSC, and EU grant QuantAlgo.

    \bibliographystyle{alphaurl}
    \bibliography{bibo-final}

    \appendix

    \section{Decision-Tree Size and Formula Size}\label{app: sizes}
    We observe in this section that the formula size of a Boolean function $f : \zone^n \to \zone$ is at most a constant times the decision-tree size of $f$. While the proof is simple, we are unaware of such a statement appearing explicitly in the literature.

    \begin{definition}[Boolean Formulas]
        A Boolean formula is a rooted binary tree, where internal nodes have in-degree 2 and are labeled with $\AND$ or $\OR$, and leaf nodes are labeled by variables or their negations ($x_i$ or $\overline{x_i}$). A Boolean formula computes a Boolean function in the natural way.
    \end{definition}
    The size of a Boolean formula is the number of nodes in it.
    Let $L(f)$ denote the minimum size of a Boolean formula computing $f$.
    \begin{claim}\label{claim: formula size dtsize}
        Let $f : \zone^n \to \zone$ be a Boolean function. Then,
        \[
        L(f) \leq 5\dtsize(f).
        \]
    \end{claim}
    \begin{proof}
        Let $\T$ be a decision tree computing $f$. If $\T$ is a single leaf, say $b$, then the size-1 formula $b$ computes $f$. Suppose $\T$ is of the form
        \begin{figure}[H]
            \centering
            \begin{tikzpicture}[yscale=-1]
                \node (T) at (-2,.5) {$\T =$};
                \node (0) at (0,0) {$x_i$};
                \node (L) at (-1,1) {$\T_L$};
                \node (R) at (1,1) {$\T_R$};
                \draw[->] (0) to node[above left] {$0$} (L);
                \draw[->] (0) to node[above right] {$1$} (R);
            \end{tikzpicture}
        \end{figure}
        The following formula can then easily be seen to compute $f$.
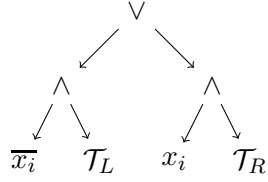
        \begin{figure}[H]
            \centering
            \begin{tikzpicture}[yscale=-1]
                \node (0) at (0,0) {$\vee$};
                \node (L) at (-1,1) {$\wedge$};
                \node (R) at (1,1) {$\wedge$};
                \node (LL) at (-1.5,2) {$\overline{x_i}$};
                \node (LR) at (-.5,2) {$\T_L$};
                \node (RL) at (.5,2) {$x_i$};
                \node (RR) at (1.5,2) {$\T_R$};
                \draw[->] (0) to node[above left] {} (L);
                \draw[->] (0) to node[above right] {} (R);
                \draw[->] (L) to node[above left] {} (LL);
                \draw[->] (L) to node[above right] {} (LR);
                \draw[->] (R) to node[above left] {} (RL);
                \draw[->] (R) to node[above right] {} (RR);
            \end{tikzpicture}
        \end{figure}
        Recursively constructing trees for $\T_L$ and $\T_R$, we obtain the recurrence
        \[
        L(f) \leq 5 + \dtsize(\T_L) + \dtsize(\T_R).
        \]
        Using induction on the size of $\T$, solving this recurrence yields the claim.
    \end{proof}

    \section{Another Weight Assignment}\label{app: another weight sqrt llogl}
    In this section we give a different weight assignment to the weights in Program~\ref{program: weights} to obtain an objective value of $O(\sqrt{\dtsize(\T) \log \dtsize(\T)})$ for a deterministic decision tree $\T$. While this bound is weaker than the second bound in Corollary~\ref{cor:ubs} by a logarithmic factor, we choose to include it since we think the weight assignment may be of independent interest. Moreover, these weights are much easier to compute than those in Section~\ref{sec: weight assignments} since the weights here have a neat closed-form expression, as opposed to the iterative construction in Section~\ref{sec: weight assignments}.

    We define the weight assignment $\cbra{W'_e : e~\textnormal{is an edge in}~\T}$ as follows. Consider any vertex $v$ of $\T$, and let $\T_v$, $\T_{v, L}$ and $\T_{v, R}$ denote the subtrees of $\T$ rooted at $v$, the left child of $v$ and the right child of $v$, respectively.
    For the two outgoing edges of $v$, define the corresponding weights as in the following figure:
    \begin{figure}[H]
        \centering
        \begin{tikzpicture}[yscale=-1]
            \node (0) at (0,0) {$v$};
            \node (L) at (-1,1) {$\T_{v, L}$};
            \node (R) at (1,1) {$\T_{v, R}$};
            \draw[->] (0) to node[above left] {$a$} (L);
            \draw[->] (0) to node[above right] {$b$} (R);
        \end{tikzpicture},
    \end{figure}
    where
    \begin{equation}\label{eqn: Prover Delayer inspired weights}
        a = \frac{1}{\log\frac{\dtsize(\T_v)}{\dtsize(\T_{v, L})}}, \qquad b = \frac{1}{\log\frac{\dtsize(\T_v)}{\dtsize(\T_{v, R})}}.
    \end{equation}

    For a leaf $\ell$ of $\T$, we have
    \begin{equation}\label{eqn: pdgame nonscaled positive witness size}
        \sum_{v \in P_\ell} \frac{1}{W'_{e(v, N(v))}} = \sum_{v \in P_\ell}\log\frac{\dtsize(\T_v)}{\dtsize(\T_{N(v)})} = \log\prod_{v \in P_\ell}\frac{\dtsize(\T_v)}{\dtsize(\T_{N(v)})}
    \end{equation}
    where the summation excludes the leaf $\ell$ and $N(v)$ denotes the child of $v$ that is on the path $P_\ell$.
    The product is telescopic and equals $\log \dtsize(\T)$.
    With the same notation and using $\overline{N(v)}$ to denote the child of $v$ that is not on the path $P_\ell$, we also have

    \begin{align}
        \label{eqn: pdgame nonscaled negative witness size}
        \sum_{e \in \overline{P_\ell}} W'_e & = \sum_{v \in P_\ell} \frac{1}{\log\frac{\dtsize(\T_v)}{\dtsize(\T_{\overline{N(v)}})}} = \sum_{v \in P_\ell} \frac{1}{\log\frac{1 + \dtsize(\T_{N(v)}) + \dtsize(\T_{\overline{N(v)}})}{\dtsize(\T_{\overline{N(v)}})}}\\
        & \leq \sum_{v \in P_\ell} \frac{1}{\log\rbra{1 + \frac{\dtsize(\T_{N(v)})}{\dtsize(\T_{\overline{N(v)}})}}}\nonumber\\
        & \leq \sum_{v \in P_\ell} \frac{1 + \frac{\dtsize(\T_{N(v)})}{\dtsize(\T_{\overline{N(v)}})}}{\frac{\dtsize(\T_{N(v)})}{\dtsize(\T_{\overline{N(v)}})}}\tag*{since $\frac{1}{\log(1 + x)} \leq \frac{1+x}{x}$ for all $x > 0$}\nonumber\\
        & \leq \sum_{v \in P_\ell}\frac{\dtsize(\T_v)}{\dtsize(\T_{N(v)})}\nonumber\\
        & = \sum_{v \in P_\ell} 1 + \frac{\dtsize(\T_v) - \dtsize(\T_{N(v)})}{\dtsize(\T_{N(v)})}\nonumber\\
        & \leq \sum_{v \in P_\ell} 1 + \dtsize(\T_v) - \dtsize(\T_{N(v)})\nonumber\\
        & = |P_\ell| + \dtsize(\T) - 1\nonumber\\
        & \leq 2\dtsize(\T)\nonumber.
    \end{align}

    Thus, this setting of weights, along with the settings
    \[
    \alpha = 2\dtsize(\T), \qquad \beta = \log \dtsize(\T),
    \]
    gives an objective value of $O(\sqrt{\dtsize(\T) \log \dtsize(\T)})$ Program~\ref{program: weights} for $\T$.
\end{document}